\pdfoutput=1
\documentclass[11pt]{article}
\usepackage{amsfonts}
\usepackage{amssymb,amsmath,amsthm}
\usepackage{graphicx}
\usepackage[dvips, paper=letterpaper, top=1in, bottom=1in, left=1in, right=1in, nohead, includefoot, footskip=.4in]{geometry}
\usepackage[authoryear,round]{natbib}
\usepackage{environ}
\usepackage{color}
\usepackage{epsfig}
\usepackage{caption}
\usepackage{verbatim}
\usepackage{subcaption}
\usepackage{float}
\usepackage{setspace}
\usepackage{enumitem}
\setlist[enumerate]{itemsep=0mm}
\usepackage{calc}
\usepackage[pdfpagemode=UseOutlines,hyperfootnotes=false, plainpages=false]{hyperref}
\usepackage{mathtools}
\usepackage{dsfont}
\usepackage[english]{babel}
\usepackage{amsthm}
\usepackage{bm}
\usepackage{breqn}
\usepackage{bigints}
\usepackage{tikz}
\usetikzlibrary{positioning}
\usepackage{nicefrac}
\usepackage{xurl}
\usepackage{algorithm}
\usepackage{algpseudocode}
\usepackage{subfiles}
\usepackage[capitalise]{cleveref}
\usepackage[flushleft]{threeparttable}
\usepackage{array,booktabs,makecell}
\usepackage{subcaption}
\usepackage{booktabs}

\usepackage[font=normalsize,labelfont=bf]{caption}
\usepackage[font=footnotesize,labelfont=bf]{subcaption}
\usepackage{sectsty}
\allsectionsfont{\normalfont\sffamily\bfseries}
\sffamily

\newtheoremstyle{theoremsansserif} 
    {\topsep}                    
    {\topsep}                    
    {\itshape}                   
    {}                           
    {\sffamily\bfseries }        
    {.}                          
    {.5em}                       
    {}  

\theoremstyle{theoremsansserif}

\newtheorem{lemma}{Lemma}
\newtheorem{corollary}{Corollary}

\newtheorem{remark}{Remark}
\newtheorem{definition}{Definition}

\newtheorem*{example*}{Example}

\newtheorem{theorem}{Theorem}

\theoremstyle{definition}
\newtheorem{exmp}{Example}

\newcommand{\R}{\mathbb{R}}
\newcommand{\I}{\mathcal{I}}
\newcommand{\A}{\mathcal{A}}

\makeatletter
\newenvironment{varsubequations}[1]
 {%
  \addtocounter{equation}{-1}%
  \begin{subequations}
  \renewcommand{\theparentequation}{#1}%
  \def\@currentlabel{#1}%
 }
 {%
  \end{subequations}\ignorespacesafterend
 }
\makeatother

\newenvironment{newsanti}{\color{black}}{}

\newcommand{\revision}[1]{{\color{black} #1}}
\newenvironment{revisionenv}{\color{black}}{}

\begin{document}
\title{\sf\textbf{Single-Leg Revenue Management with Advice}}

\author{
\sf Santiago R. Balseiro\\
\sf Columbia University \\
\small\texttt{srb2155@columbia.edu}
\and
\sf Christian Kroer \\
\sf Columbia University\\
\small\texttt{christian.kroer@columbia.edu}
\and
\sf Rachitesh Kumar\\
\sf Columbia University \\
\small\texttt{rk3068@columbia.edu}}

\date{\vspace{1em}
\sf This version: \today}

\maketitle

\allowdisplaybreaks

\begin{abstract}
	Single-leg revenue management is a foundational problem of revenue management that has been particularly impactful in the airline and hotel industry: Given $n$ units of a resource, e.g. flight seats, and a stream of sequentially-arriving customers segmented by fares, what is the optimal online policy for allocating the resource. Previous work focused on designing algorithms when forecasts are available, which are not robust to inaccuracies in the forecast, or online algorithms with worst-case performance guarantees, which can be too conservative in practice. In this work, we look at the single-leg revenue management problem through the lens of the algorithms-with-advice framework, which attempts to harness the increasing prediction accuracy of machine learning methods by optimally incorporating advice about the future into online algorithms. In particular, we provide an online algorithm that attains every point in the Pareto frontier between consistency (performance when advice is accurate) and competitiveness (performance when advice is inaccurate) for \emph{every advice}. We also study the class of protection level policies, which is the most widely-deployed technique for single-leg revenue management: we provide an algorithm to incorporate advice into protection levels that optimally trades off consistency and competitiveness. Moreover, we numerically evaluate the performance of these algorithms on synthetic data. We find that our algorithm for protection level policies performs remarkably well on most instances, even if it is not guaranteed to be on the Pareto frontier in theory. Our results extend to other unit-cost online allocations problems such as the display advertising \revision{and the multiple secretary problem together with more general variable-cost problems such as the online knapsack problem.}
\end{abstract}


\thispagestyle{empty}
\pagebreak
\setcounter{page}{1}

\setstretch{1.5}
\pagebreak

\graphicspath{ {Images/} }

\section{Introduction}

The field of revenue management, one of the pillars of operations research, got its start with the airline industry in the twentieth century \citep{talluri2004theory}. Since then, it has found use in a variety of industries, covering the entire gamut from retail to hospitality. The goal of revenue management is to design price and quantity control policies that optimize the revenue of a firm. In this work, we will focus on quantity control. In particular, we consider a single-resource unit-cost resource allocation problem in which the decision maker wants to optimally allocate a limited inventory of a single resource to sequentially arriving requests. Each request consumes one unit of inventory, and generates a reward. For historical reasons, the terminology of revenue management is tailored to the airlines industry, and we continue with this convention in this work, but it is worth noting that the model and results apply more generally. For example, our results apply to the so-called display ad problem, in which impressions need to be assigned to advertisers to maximize clicks~\citep{feldman2010online}, or the multiple secretary problem, in which applicants arrive sequentially and the best candidates need to be hired~\citep{kleinberg2005multiple}. \revision{Moreover, we can also cover variable-cost problems such as the online knapsack problem~\citep{zhou2008budget} or budget-constrained bidding in repeated auctions~\citep{balseiro2019learning}.}


Consider an airline that operates a flight with $n$ economy seats. The seats in the economy cabin are demanded by a variety of customer types, which motivates airlines to offer different fare classes, each of which is designed to cater to a different market segment. For example, airlines often offer Basic Economy fares for leisure travelers who are price-sensitive. These low-fare tickets do not afford the holder any perks like seat selection, luggage check-in, upgrade eligibility, extra miles, priority boarding etc. On the other end of the spectrum are Full Fare Economy tickets that come with all of the aforementioned perks and are designed to be frequently available for late bookings for business travelers who have to travel at short notice. Given this collection of fare classes (which we assume to be fixed), how should an airline control the number of seats sold to customers from different fare classes in order to maximize revenue? This is referred to as the single-leg revenue management problem.

The crux of the single-leg revenue management is captured by the following trade-off: If the airline sells too many seats to customers from lower fare classes, then it will not be able to sell to higher-fare-class customers that might arrive later, and if the airline protects too many seats for higher fare-class customers, it will lose revenue from the lower fare classes if the demand for higher-fare classes never materializes. In the absence of any information about the fare classes of customers that will arrive, this problem falls under the paradigm of online algorithms and competitive analysis. This is the approach taken by \citet{ball2009toward}, who characterized the optimal performance (in terms of competitive ratio) that any policy can achieve. In contrast, the vast majority of past work on single-leg revenue management assumes that accurate distributional forecasts are available about the customers that will arrive, and then proceeds to characterize the optimal policy in terms of the forecasts (see \citealt{gallego2019revenue} for a recent overview). Often, additional assumptions like low-before-high (customers belonging to lower-fare classes arrive first) are also required for these results.

It would come as no surprise that policies which leverage the forecast perform much better than the policy of \citet{ball2009toward} when the sequence of customers is consistent with the forecast, but lose all performance guarantees when this is not the case. This work aims to achieve the best of both worlds by marrying the robustness of competitive analysis with the superior performance that can be achieved with forecasts. Robustly incorporating forecasts into algorithms is of fundamental importance in today's data-driven economy, with the rising adoption of machine-learning-based prediction algorithms. To do so, we assume that we have access to advice about the sequence of customers that will arrive, but we do not make any assumptions about the accuracy of this advice. We develop algorithms that perform well when the advice is accurate, while maintaining worst-case guarantees for all sequences of customers. {\newsanti Because our algorithms take the advice as a black box, they can harness the increasing prediction power of machine learning methods.} Our approach falls under the framework of Algorithms with Advice, which has found wide application of late (see \citealt{mitzenmacher2020algorithms} for a recent survey).

Before moving onto our contributions, we briefly discuss the centerpiece of single-leg revenue management theory and practice: protection level policies (also called booking limit policies), which play a vital role in our results. It is a class of policies parameterized by protection levels, one for each fare class. A protection level for a fare class is a limit on the number of customers that are accepted with fares lying below that fare class. A protection level policy accepts a customer if doing so does not violate any of the protection levels and rejects her otherwise. The motivation behind these limits (and the inspiration for the name `protection level policies') is the need to protect enough capacity for higher fare classes and prevent lower fare classes from consuming all of the capacity. These policies have desirable theoretical properties in a variety of models and are widely deployed in practice (see \citealt{talluri2004theory} for a detailed discussion).

\subsection{Main Contributions}

As is often the case for the Algorithms with Advice framework, we assume that the advice contains the minimum amount of information necessary to achieve optimal performance when the advice is accurate. For the single-leg revenue management problem, this translates to the top $n$ fares that will arrive, specified as a frequency table of fare classes (e.g. 60 Basic Economy customers and 40 Full Fare customers) without any information about the order. {\newsanti Our choice of advice has the advantage of being easy to estimate because of its low sample complexity.} Given an advice, we consider the following competing goals: (i) \emph{Consistency:} The competitive ratio of the algorithm when the advice is accurate, i.e., the instance conforms with the advice; (ii) \emph{Competitiveness:} The worst-case competitive ratio over all sequences of customers, regardless of conformity to the advice.

Since different firms might have differing levels of confidence in the advice they receive (from machine-learning systems or humans), we are motivated to study the frontier that captures the trade-off between consistency and competitiveness. In particular, we allow the required level of competitiveness to be specified as an input alongside with the advice, and then attempt to maximize the level of consistency that can be attained while maintaining this level of competitiveness. We note that the choice of competitiveness as input is arbitrary and all of our algorithms can be modified to take the required level of consistency as input instead. Before stating our results, we describe an example that illustrates their flavor.

\begin{figure}[t]
	\centering
	\includegraphics[width = 0.5\linewidth]{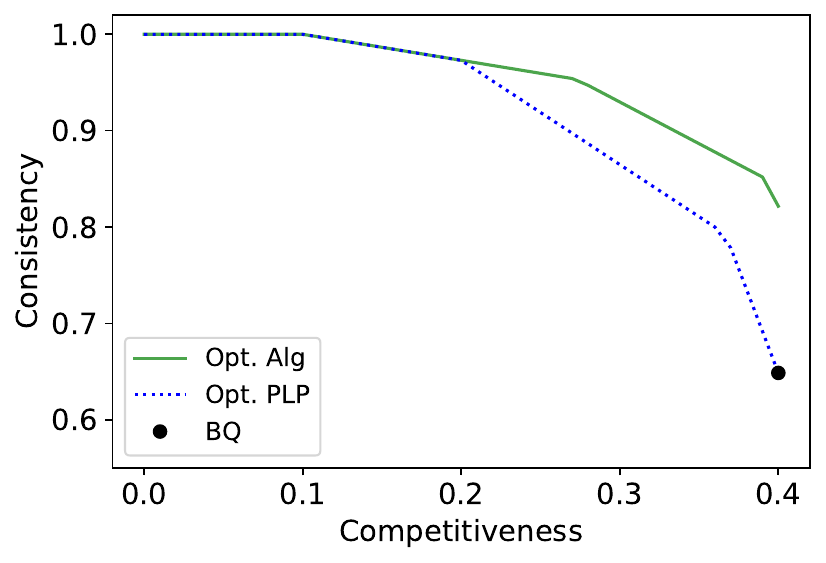}
	\caption{The tradeoff between consistency and competitiveness for Example~\ref{example:intro}. The solid line represents the performance of the LP-based optimal algorithm (Algorithm~\ref{alg:opt-alg}) and the dotted line represents the performance of the optimal protection level policy (Algorithm~\ref{alg:opt-protec}). We use the performance of the policy of \citet{ball2009toward}, which maximizes competitiveness, as a benchmark. Here 0.4 is the maximum level of competitiveness that can be achieved by any online algorithm.}
	\label{fig:intro}
\end{figure}

\begin{exmp}\label{example:intro}
	Consider a plane with $100$ economy seats and fare classes given by $\{\$100, \$200, \$400, \$800\}$. The airline is advised that, of the customers predicted to arrive, 20 customers would want the \$200 fare class, 60 customers would want the \$400 fare class and 10 customers would want the \$800 fare class, while the remaining customers (assumed to be at least 10 in number) would want the \$100 fare class. Suppose the airline is not completely confident about this advice and would like to optimally tradeoff consistency and competitiveness. In this work, we show that the consistency-competitiveness Pareto frontier that captures the optimal way to perform this tradeoff for every advice can be computed efficiently using a Linear Program (\ref{LP}), and give an algorithm (Algorithm~\ref{alg:opt-alg}) that achieves this optimal performance. We depict the Pareto frontier for the advice above with a solid line in Figure~\ref{fig:intro}. If we are restricted to using protection level policies, then we describe another algorithm (Algorithm~\ref{alg:opt-protec}) that yields the optimal tradeoff. We depict the performance of Algorithm~\ref{alg:opt-protec} with a dotted line in Figure~\ref{alg:opt-protec}. Our results show that consistency can be significantly increased by sacrificing only a small amount of competitiveness.
\end{exmp}

\textbf{The Consistency-Competitiveness Pareto Frontier.} We construct an LP, and develop an optimal algorithm based on it, that completely characterizes the Pareto frontier of consistency and competitiveness. More precisely, we give an efficient LP-based optimal algorithm (Algorithm~\ref{alg:opt-alg}) which, given an advice and a required level of competitiveness, achieves the highest level of consistency on that advice among all online algorithms that satisfy the required level of competitiveness. {\newsanti We remark that most algorithms in the literature do not attain the best level of consistency for every advice, but instead, are only shown to be optimal for the worst-case advice. In contrast, our LP-based optimal algorithm attains the best possible level of competitiveness for \emph{all advice and levels of competitiveness}.} In other words, for every advice, we construct an LP to efficiently compute the consistency-competitiveness Pareto frontier, and also develop an online algorithm whose performance lies on this Pareto frontier. We achieve this through the following steps: (i) First, we assemble a collection of hard customer sequences for the given advice; (ii) Then, we construct an LP that aims to maximize consistency while maintaining the required level of competitiveness on these hard instances; (iii) Finally, we use the solution of the LP to construct a collection of protection levels, and optimally switch between these protection levels, to attain the highest possible level of consistency, while attaining the required level of competitiveness on \emph{all} customers sequences. The crux of the proof lies in guessing the right collection of sequences to distill the fundamental tradeoffs of the problem, and then proving that these are in fact the hardest sequences, i.e., our online algorithm which is designed to do well on these hard sequences performs well on all sequences.

\textbf{Optimal Protection Level Policy.} In recognition of the central role of protection level policies in single-leg revenue management theory and practice, we also characterize the optimal way to incorporate advice into protection level policies. We emphasize that, even though our optimal algorithm makes use of protection levels to make decisions, it is not a protection level policy because it may switch between protection levels. This switch between protection levels comes at the cost of some desirable practical properties of protection level policies like monotonicity (never rejecting a customer from a certain fare class and then accepting a customer from the same fare class that arrives later) and being oblivious to the fare class of the customer before making the accept/reject decision. However, we show that this switch is necessary for optimality: protection level policies can be sub-optimal in their consistency-competitiveness trade-off. Nonetheless, the significance of protection level policies to the theory and practice of revenue management warrants our investigative efforts. We develop an algorithm (Algorithm~\ref{alg:opt-protec}) that takes as input an advice and required level of competitiveness, and outputs protection levels which correspond to the protection level policy that maximizes consistency in the set of all protection level policies that satisfy the required level of competitiveness. To rephrase, regardless of the advice, no protection level policy can achieve higher consistency and competitiveness than the one based on Algorithm~\ref{alg:opt-protec}.

\revision{\textbf{Robustness.} Since no prediction algorithm can be completely accurate, it is also important for any algorithm to achieve good performance when the sequence of fares is ``close'' to the advice. We use an appropriately modified $\ell_1$ norm to define a distance between the realized fare sequences and the advice, and show that the performance of our algorithms degrades gracefully as a function of this distance. In particular, we show that the performance of protection level policies degrades linearly as a function of the distance. Moreover, we describe a relaxed version of our optimal algorithm (Algorithm~\ref{alg:robust-opt-alg}) whose performance also degrades linearly with the distance in a neighborhood of the advice.}

\textbf{Numerical Experiments.} Finally, we also run numerical experiments on synthetic data with the aim of (i) comparing the performance of our optimal LP-based algorithm (Algorithm~\ref{alg:opt-alg}) and our optimal protection level policy (Algorithm~\ref{alg:opt-protec}); and (ii) comparing the average performance of our algorithms when the sequence of customers is drawn from a distribution centered on the advice. We find that protection level policies are optimal for most types of advice, and that the essence of its sub-optimality is captured by our bad example (Example~\ref{example:bad}). Moreover, we find a graceful degradation in the performance of our algorithms as a function of the noise in the distribution that generates the sequences. We also find that our algorithms continue to perform better than that of \citet{ball2009toward} even under high levels of noise.

\subsection{Additional Related Work}

\textbf{Single-Leg Revenue Management.} The single-leg revenue management problem has a long history in revenue management. \citet{littlewood2005special} characterized the optimal policy for the two-fare class setting under known-stochastic customer arrival with the LBH (low before high, i.e., customers arrive in increasing order of fares) assumption. \citet{brumelle1993airline} extended the results to multiple fare classes, under the additional assumption of independence across fare classes, via a dynamic programming formulation resulting in a protection level policy that is optimal. \citet{lee1993model}, \citet{robinson1995optimal} and \citet{lautenbacher1999underlying} dispense with the LBH assumption and characterize the optimal policy in this dynamic setting. All of the aforementioned works assume that the stochastic process governing customer arrival is completely known. \citet{van2000revenue}, \citet{kunnumkal2009stochastic} and \citet{huh2006adaptive} relax this assumption by considering a setting in which LBH demand is drawn from some stationary unknown distribution repeatedly. They give an adaptive procedure for updating the protection levels using samples from this distribution that converges to the optimal protection levels. The stationarity assumption allows them to learn the optimal protection levels perfectly for future demand. As discussed earlier, \citet{ball2009toward} broke from the tradition of stochastic assumptions and looked at the problem through the lens of competitive analysis. \citet{lan2008revenue} extended their work by allowing for known upper and lower bounds on the number customers from each fare class. Under the assumption that these bounds always hold, they develop algorithms which achieve the optimal competitive ratio. Importantly, their algorithm loses its guarantees when the bounds are violated. Like \citet{lan2008revenue}, we also develop a factor-revealing LP. But, unlike their LP, which is simple enough to be solved in closed form, our LP solves a complex optimization problem based on more intricate hard instances and does not admit a closed-form solution. \citet{ma2021policies} generalize the single-leg model by allowing the firm to price the fare classes, and show that this can be done without any loss in the competitive ratio. We refer the reader to the comprehensive books by \citet{talluri2004theory} and \citet{gallego2019revenue} for a detailed discussion on single-leg revenue management. \revision{Recent work has also studied the single-leg revenue management problem under other models which go beyond the classical adversarial and stochastic analysis. \citet{hwang2021online} develop online learning algorithms for a model with two fare classes in which the sequence of fares is selected by an adversary, but a randomly selected subset of that sequence arrives in random order. Each customer joins the random-order subset independently with the same probability and this random component of the input sequence allows their algorithm to partially learn the sequence. \citet{aouad2022nonparametric} study the more general stochastic matching problem in which the number of customers of each type are drawn from some distribution, and then these customers arrive in an adversarial or random order. \citet{golrezaei2022online} also study a setting with two types of fare classes in which the decision maker does not know the reward of each fare class. They posit the existence of a test period in which each customer is sampled independently with the same probability and reveals her reward to the user. Their algorithm uses this sample information to obtain asymptotically optimal performance.}

\textbf{Algorithms with Advice.} Motivated by the ubiquity of machine-learning systems in practice, a recent line of work looks at improving algorithmic performance by incorporating predictions. Here, the goal is to avoid making assumptions about the quality of the predictions and instead design algorithms that perform better when the prediction is accurate while maintaining reasonable worst-case guarantees. This framework has predominantly been applied to online optimization problems like caching \citep{lykouris2018competitive, rohatgi2020near}, ski-rental and online scheduling \citep{purohit2018improving}, the secretary problem \citep{antoniadis2020secretary,dutting2021secretaries}, online ad allocation \citep{mahdian2012online} and queueing \citep{mitzenmacher2021queues}. While our problem can be thought of as an online allocation problem, the results of \citep{mahdian2012online} are not directly applicable because they assume that rewards are proportional to resource consumption, which does not hold for the revenue management problem. {\newsanti The classical secretary problem (and its multi-secretary generalizations) share the basic setup with the single-leg revenue management problem---they are both single-resource unit-cost resource allocation problems. The difference being that, unlike the stochastic assumptions made in the secretary literature, the single-leg revenue management model considered here only requires the rewards to belong to some known finite set and allows for them to be adversarially chosen. This makes the existing results on the secretary problem with advice \citep{antoniadis2020secretary,dutting2021secretaries} incomparable to ours.} Advice has also been used to improve data structures \citep{mitzenmacher2019model} and run-times of algorithms \citep{dinitz2021faster}. We refer the reader to \citet{mitzenmacher2020algorithms} for a more detailed discussion.

\textbf{Other Work.}  The single-leg revenue management problem is a special type of online packing problem with one resource, unit costs and varying rewards coming from a known set. Under adversarial input, a constant competitive ratio cannot be obtained for online packing in general, but a long line of work shows that it can be achieved for important special cases like online matching \citep{karp1990optimal}, the Adwords problem \citep{mehta2007adwords} and the Display Ads problem with free disposal \citep{feldman2009online}. On the other hand, near optimal performance can be obtained for stochastic or random order input under mild assumptions \citep{devanur2009adwords, feldman2010online, agrawal2014dynamic, alaei2012online}. See \citet{mehta2013online} for a survey on online allocation. The Adwords problem, which is particular important for the online advertising industry, has also been studied in a variety of models that interpolate between the adversarial and stochastic settings \citep{mirrokni2012simultaneous, esfandiari2015online}.

\section{Model}

We consider a single-resource unit-cost resource allocation problem: a firm (airline) has $n$ units of a resource (seats on a plane) it wants to allocate to sequentially arriving requests (customers). Each request consumes one unit of the resource and generates a reward. We assume that the reward can take $m$ possible values (fare classes), which we denote by $0 < f_1 < f_2 < \dots < f_m$. To simplify notation, we include a special zero request with $f_0 = 0$. At each time step $t$, a request arrives, whereupon its reward is revealed to the firm and an irrevocable accept/reject decision is made. We only assume that the firm knows the set of possible rewards $F \coloneqq \{f_1, \dots f_m\}$, and do not require the firm to know the total number of requests that will arrive.


\subsection{Applications}

Before proceeding further with the model, we discuss its applicability to internet advertising markets, the online knapsack problem, the multi-secretary problem, and the airline industry. In the display ad problem, which has received considerable attention due to its prominence in internet advertising, advertisers enter into contracts with websites to display a certain number of ads~\citep{feldman2010online}. Here the capacity $n$ is the number of impressions guaranteed by the contract and when a visitor arrives the website needs to decide whether to display an ad to the visitor, which counts against the advertiser's contract and consumes one unit of this resource. The rewards are the appropriately discretized click-through-rate (CTR) estimates for showing an ad to the visitor. The goal of the website is to maximize the total CTR while satisfying the contractual obligations.\footnote{\revisionenv Even though reservation contracts have requirements, in practice they are treated as packing problems~(see, e.g., \citealt{mehta2013online}). A common approach is for publishers to penalize shortfalls by adding penalties to the objectives, which can be incorporated to our model by adjusting fares down.}

\revision{Our framework can also be used to model the online knapsack problem~\citep{zhou2008budget} when the possible number of item types is finite. Without loss of generality, we can assume that the weight of each item is a non-negative integer because such an assumption can always be satisfied with rescaling. Given an instance of the online knapsack problem, we can construct an instance of the single-leg revenue management problem as follows: If an item has weight $w \in \mathbb{N}$ and reward $r$, replace it with $w$ items, each with weight 1 and reward $r/w$. Thusm if the possible weights are $1 = w_1 < w_2 < \dots < w_a$ and the possible item rewards are $r_1 < \dots < r_b$, then the corresponding single-leg revenue management problem with fares $\{r_i/w_j \mid i \in [b];\ j \in [a]\}$ has the same worst-case competitive ratio. With this reduction, the guarantees for our algorithms which incorporate advice also continue to hold for the online knapsack problem (see Appendix~\ref{appendix:knapsack} for details). Moreover, \citet{zhou2008budget} show that the online knapsack problem captures budget-constrained bidding in repeated auctions, thereby also bringing the latter under the purview of our framework.}

The multi-secretary problem \citep{kleinberg2005multiple} captures the setting in which a company wishes to hire $n$ employees, each of whom takes up one position and generates varying amounts of productivity (reward). It is assumed that the employee is interviewed and the associated reward is revealed to the company before making the accept/reject decision. Our model captures the adversarial-version of the multi-secretary problem where no assumption is made about the rewards other than membership in some known set $F$.

Finally, as discussed earlier, the origin of our model lies in the airline industry, where the fare classes $F$ are decided in advance, customers arrive online, and accept/reject decisions need to be made to maximize revenues. Moreover, it is a standard assumption in the single-leg revenue management literature that the different fare classes are designed to perfectly segment the market through their perks \citep{talluri2004theory}: customers belong to a fare class and do not substitute between them. In the rest of the paper, we continue with tradition and use the terminology of the airline industry and the single-leg revenue management problem to describe our model. Henceforth, requests will be referred to as customers and the possible rewards $F$ as fare classes or fares.

\subsection{Advice, Online Algorithms, and Performance Metrics}

We next present some definitions that will be used in our analysis.

\begin{definition}
	An instance $I$ of the single-leg revenue management problem is a variable-length sequence of customer fares $s_1, \dots, s_T$, where $T \in \mathbb{Z}_+$ and $s_i \in \{f_1, f_2, \dots, f_m\}$ for all $i \in [T]$. We denote the set of all instances by $\mathcal{U}$.
\end{definition}

As is customary in competitive-ratio analysis, we will measure the performance of an online algorithm by comparing it to the performance of the clairvoyant optimal.

\begin{definition}
	For an instance $I = \{s_t\}_{t=1}^T$ of the single-leg revenue management problem, $Opt(I)$ represents the maximum revenue one can obtain if one knew the entire sequence $\{s_t\}_{t=1}^T$ in advance, i.e., $Opt(I)$ is the sum of the highest $n$ fares in $I$ if $T \geq n$ and the sum of all fares if $T<n$. 	
\end{definition}

The classical work of \citet{ball2009toward} assumes that, at time $t$, the firm has no information about the future sequence of fares $s_{t+1}, \dots, s_T$, and then goes on to characterize the worst-case competitive ratio (the ratio of the revenue of the firm to that of a clairvoyant entity that knows the entire instance) over all instances. In today's data-driven world, their assumption about the complete lack of information about future fares can be too pessimistic. We relax it and assume that the firm has access to an oracle that provides the firm with predictions about the fares that will arrive. In practice, this oracle is often a machine-learning model trained on past data about fares. We capture this through the notion of advice as used in the algorithms-with-advice framework (see \citealt{mitzenmacher2020algorithms} for a survey).

\begin{definition}
	An advice is an element of the set $\A \coloneqq \{A \in \mathbb{N}^m \mid \|A\|_1 = n\}$ which represents the top $n$ highest fares that are predicted to arrive. In particular, an advice $A \in \A$ is interpreted to mean that we will receive an instance $I \in \mathcal{U}$ for which it would be possible to pick $A_i$ customers with fare $f_i$ for all $i \in [n]$, and doing so would yield the optimal revenue of $Opt(I)$. In light of this, we will use $Opt(A)$ to denote $\sum_{i=1}^m f_i A_i$ for all $A \in \A$.
\end{definition}

Before proceeding further, we note some important properties of the advice:
\begin{itemize}
	\item It does not take into account the order of the fares and is limited to the frequency of each fare type. Moreover, it only requires information about the top $n$ highest fares. 
	\item Each element $A \in \A$ contains just enough information to compute the optimal solution for any instance that conforms with the advice: pick $A_i$ customers with fare $f_i$ for all $i \in [n]$.
\end{itemize}
\revision{Our choice of advice drastically reduces the size of the space of possible predictions, allowing for efficient learning from past data. For example, firms could train machine learning models that map contexts (season, flight information, economic trends, etc.) to a prediction on the number of customers, which translates to a regression problem for each fare class. Moreover, compared to other choices of advice such as the arrival order of customers, our advice is permutation invariant, which leads to more robust algorithms that are less sensitive to perturbations in the data.}

In this work, we consider online algorithms that take as input an advice, and then make accept/reject decisions for sequentially arriving customers in a non-clairvoyant fashion. The goal is to perform well when the arriving customers conform to the advice (consistency), while maintaining worst-case guarantees for the case when the advice is not a good prediction of customer fares (competitiveness). Like \citet{ball2009toward}, we allow the online algorithm to accept fractional customers.

\begin{definition}
	An online algorithm $P$ (possibly randomized) that incorporates advice takes as input an advice $A \in \A$, an instance $\{s_i\}_{i=1}^{t-1}$ as history and a fare from $s_{t} \in \{f_1, \dots, f_m\}$ as the fare type of the current customer; and outputs a fractional accept quantity $P_t(A, \{s_i\}_{i=1}^{t-1}, s_{t}) \in [0,1]$ that satisfies the capacity constraint, i.e., almost surely it satisfies $\sum_{i=1}^t P_t(A, \{s_j\}_{j=1}^{i-1}, s_{i}) \leq n\,.$
\end{definition}

\begin{remark}
	Since all of the algorithms discussed in this work incorporate advice, we will refer to online algorithms that incorporate advice simply as online algorithms.
\end{remark}
\begin{remark}\label{remark:rounding-error}
	Even though we allow online algorithms to fractionally accept customers, all of the online algorithms we propose in this work satisfy the following property: If we run the algorithms with the modified capacity of $n-2m$ and completely accept the fractionally accepted customers then the number of acceptances does not exceed $n$ and degradation in the performance of the online algorithm is $O(m/n)$. Since most applications (including the airline industry) satisfy $m \ll n$, our algorithms have negligible rounding error. See Appendix~\ref{appendix:rounding-error} for details.
\end{remark}

Next, we rigorously define the performance metrics needed to evaluate online algorithms that incorporate advice, beginning with the definition of consistency, which captures the performance of an online algorithm when the instance conforms to the advice.

\begin{definition}
	For advice $A \in \A$, let $\ell$ be the smallest index such that $A_\ell \geq 1$. Define $S(A) \coloneqq \{ I = \{s_t\}_{t=1}^T \in \mathcal{U} \mid \sum_{t=1}^T \mathds{1}(s_t = f_i) = A_i\ \forall i \geq \ell+1;\ \sum_{t=1}^T \mathds{1}(s_t = f_\ell) \geq A_\ell \}$ to be the set of all instances that conform to the advice $A$. An online algorithm $P$ is $\beta$-consistent on advice $A$ if, for any instance $I = \{s_t\}_{t=1}^T \in S(A)$, it satisfies $\mathbb{E}_P\left[\sum_{t=1}^T P_t(A, \{s_i\}_{i=1}^{t-1}, s_t) \cdot s_t\right] \geq \beta \cdot Opt(A)\,.$
\end{definition}

The following definition captures the performance of online algorithms on all instances, regardless of their conformity to the advice.

\begin{definition}
	An online algorithm $P$ is $\gamma$-competitive if for any advice $A \in \A$ and instance $I = s_1, s_2, \dots, s_T$; we have $\mathbb{E}_P\left[\sum_{t=1}^T P_t(A, \{s_i\}_{i=1}^{t-1}, s_t) \cdot s_t \right] \geq \gamma \cdot Opt(s_1, \dots, s_T)\,.$
\end{definition}

%
%

\citet{ball2009toward} showed that no online algorithm can be $\gamma$-competitive for any $\gamma$ strictly greater than $c(F) \coloneqq \left[\sum_{i=1}^m (1 - (f_{i-1}/f_{i})) \right]^{-1}$. Hence, every online algorithm is $\gamma$-competitive for some $\gamma \in [0, c(F)]$.

The central goal of this work is to study the multi-objective optimization problem over the space of online algorithms which entails maximizing both consistency and competitiveness. In particular, for every advice, we characterize the Pareto frontier of consistency and competitiveness, thereby capturing the tradeoff between these two incongruent objectives. We conclude this section with the definition of protection level policies, which are a class of online algorithms that play a critical role in our results (and single-leg revenue management in general).

\subsection{Protection Level Policies} \label{subsec:protection-level-policy}

Protection level policies are the method-of-choice for single-leg revenue management in the airline industry, and have garnered extensive attention in the revenue management literature~\citep{talluri2004theory, ball2009toward}. It is a form of quantity control parameterized by protection levels $0 \leq Q_1 \leq Q_{2} \leq \dots \leq Q_m = n$: the firm accepts at most $Q_i$ customers with fares $f_i$ or lower. The idea stems from the need to preserve capacity for higher paying customers that might arrive in the future. More precisely, on an instance $s_1, \dots s_T$ where $s_t = f_j$ for some $j \in [m]$, the fraction of $s_t$ accepted by the protection level policy with protection levels $Q = (Q_1, \dots, Q_m)$ can be iteratively defined as
\begin{align*}
	x_t = \max\left\{ x \in [0,1]\ \biggr\lvert\ x + \sum_{i=1}^{t-1} x_i\mathds{1}(s_i \leq f_k) \leq Q_k \text{ for all } k \geq j \right\}\,.
\end{align*}

Our algorithm (Algorithm~\ref{alg:opt-alg}) will use protection level policies as subroutines to attain the optimal level of consistency for a given level of competitiveness. We note, however, that this does not make Algorithm~\ref{alg:opt-alg} a protection level policy, since it may switch between different protection levels. Moreover, in reverence to their role in revenue management theory and practice, we will also characterize the protection level policy that attains the optimal level of consistency for a given level of competitiveness among all protection level policies (Algorithm~\ref{alg:opt-protec}). We conclude this section with an example to build some intuition.

\begin{exmp}
	Let $n = 3$, $F = \{f_1, f_2, f_3\}$ and $Q_1 = 1, Q_2 = 2, Q_3 = 3$. Consider the execution of the protection level policy parameterized by $Q$ on the sequence $f_2, f_2, f_1, f_3$. Then, since $Q_2 = 2$, we accept the first two customers with fare $f_2$. At the third step, we reject the customer with fare $f_1$ because we have already accepted two customers with fare $f_2$ or below, and thus accepting it would violate the second protection level $Q_2 = 2$. Note that we reject this customer despite the fact that we have not accepted any customer with fare $f_1$ and $Q_1 = 1$. Finally, we accept the customer with fare $f_3$ in the fourth time step. Now, contrast this with the sequence $f_2, f_1, f_2, f_3$. In this case, we would accept the first customer with fare $f_2$ like before. We would also accept the second customer with fare $f_1$ because it does not violate either $Q_1$ or $Q_2$. In the third step, we would reject the customer with fare $f_2$ because we have already accepted two customers with fare $f_2$ or below. Finally, we accept the fourth customer with fare $f_3$.
\end{exmp}

\section{The Consistency-Competitiveness Pareto Frontier} \label{sec:pareto}

The goals of consistency and competitiveness can be at odds with each other. Therefore, depending on the level of confidence in its prediction model, different firms will want to target different levels of consistency and competitiveness. Nonetheless, it would always be desirable to use an online algorithm that optimally trades off consistency and competitiveness. In other words, we would like online algorithms whose performance lies on the consistency-competitiveness Pareto frontier.

To achieve this, we will describe an online algorithm that, given an advice $A$ and a desired level of competitiveness, achieves the optimal level of consistency on advice $A$ while maintaining $\gamma$-competitiveness.


\begin{definition}
	Fix an advice $A$ and a level of competitiveness $\gamma \in [0, c(F)]$. Let $\mathcal{P}_\gamma$ denote the set of all online algorithms that are $\gamma$-competitive. Then, the optimal level of consistency for advice $A$ under the $\gamma$-competitiveness requirement is defined as follows:
\begin{align*}
	\beta(A, \gamma) \coloneqq \max_{P \in \mathcal{P}_\gamma} \min_{\{s_t\}_{t=1}^T \in S(A)} \frac{\mathbb{E}_P\left[\sum_{t=1}^T P_t(A, \{s_i\}_{i=1}^{t-1}, s_t) \cdot s_t\right]}{Opt(A)}\,.
\end{align*}
\end{definition}

In subsection~\ref{subsec:hard-instances}, we discuss a collection of instances that capture the core difficulty that every non-clairvoyant online algorithm faces in maximizing consistency. Using this collection of hard instances, in subsection~\ref{subsec:upper-bound} we establish an LP-based upper bound on $\beta(A, \gamma)$. Then, in subsection~\ref{subsec:optimal-alg}, we define an algorithm that achieves a level of consistency that matches this LP-based upper bound, thereby completely characterizing $\beta(A, \gamma)$. For the rest of this section, fix an advice $A$ and a level of competitiveness $\gamma \in [0, c(F)]$. Moreover, let $\ell = \min\{ i \in [m] \mid A_i \geq 1\}$ be the lowest fare that forms a part of $A$.

\subsection{Hard Instances}\label{subsec:hard-instances}

In our optimal algorithm, which we discuss in the next subsection, we will use protection level policies as subroutines to make accept/reject decisions. It is not too difficult to see that, given a collection of customers, every protection level policy earns the least amount of revenue when the customers arrive in increasing order of fares (see Lemma~\ref{lemma:protec-inc-order} in Appendix~\ref{appendix:opt-alg} for a formal proof). Looking ahead, we will define hard instances in which the customers arrive in two phases, with the fares satisfying an increasing order in each phase.

First, observe that any $\beta$-consistent algorithm must obtain a revenue of $\beta \cdot Opt(I)$ on any instance $I \in S(A)$ that conforms with the advice.  To capture this, we define $I(A)$ to be the ``largest" instance that conforms to the advice $A$, i.e., $I(A)$ is the instance that consists of $n$ customers of fare type $f_j$ for $j \leq \ell$ and $A_i$ customers of fare type $f_j$ for $j > \ell$ arriving in increasing order of fares (see Figure~\ref{fig:advice_inst}). More precisely, set $T = n - A_\ell + (m - \ell + 1) n$ and define the instance $I(A) = \{s(A)_t\}_{t=1}^{T}$ as
\begin{itemize}
	\item $s(A)_t = f_j$ if $(j-1)n < t \leq jn$ for some $j \leq \ell$ and
	\item $s(A)_t = f_k$ if $\ell n + \sum_{i = \ell + 1}^{k-1} A_i < t \leq \ell n + \sum_{i = \ell+1}^{k} A_i$ for some $ k > \ell$\,.
\end{itemize}
 Note that $I(A) \in S(A)$ and $Opt(I(A)) = Opt(A) = \sum_{i=\ell}^m f_i A_i$, i.e., picking $A_i$ customers of fare type $f_i$ is optimal.

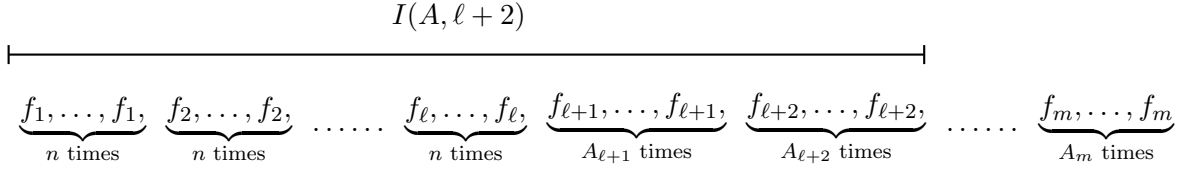
\begin{figure}[t!]
\centering
\begin{tikzpicture}
	\node(x1) {$\underbrace{f_1, \ldots, f_1,}_{n \text{ times}}$};
    \node(x2) [right = 0cm of x1] {$\underbrace{f_2, \ldots, f_2,}_{n \text{ times}}$};
    \node(x3) [right = 0cm of x2] {$\ldots \ldots$};
    \node(x4) [right = 0cm of x3] {$\underbrace{f_\ell, \ldots, f_\ell,}_{n \text{ times}}$};
    \node(x5) [right = 0cm of x4] {$\underbrace{f_{\ell+1}, \ldots, f_{\ell+1},}_{A_{\ell+1} \text{ times}}$};
    \node(x6) [right = 0cm of x5] {$\underbrace{f_{\ell+2}, \ldots, f_{\ell+2},}_{A_{\ell+2} \text{ times}}$};
    \node(x7) [right = 0cm of x6] {$\ldots \ldots$};
    \node(x8) [right = 0cm of x7] {$\underbrace{f_{m}, \ldots, f_{m}}_{A_{m} \text{ times}}$};
    \draw[thick,|-|] (-1,1) -- (11.2,1);
    \node at (5, 1.5) {$I(A,\ell+2)$};
\end{tikzpicture}
\caption{$I(A)$ is the ``largest'' instance that conforms with the advice.}
\label{fig:advice_inst}
\end{figure}

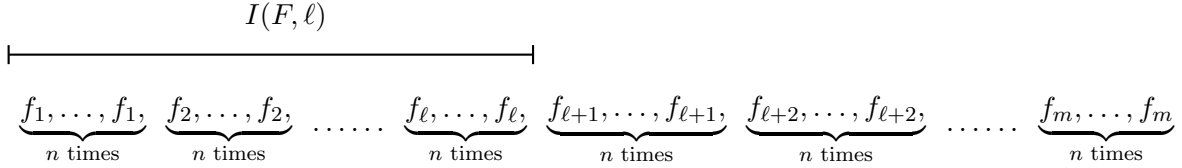
\begin{figure}[t!]
\centering
\begin{tikzpicture}
	\node(x1) {$\underbrace{f_1, \ldots, f_1,}_{n \text{ times}}$};
    \node(x2) [right = 0cm of x1] {$\underbrace{f_2, \ldots, f_2,}_{n \text{ times}}$};
    \node(x3) [right = 0cm of x2] {$\ldots \ldots$};
    \node(x4) [right = 0cm of x3] {$\underbrace{f_\ell, \ldots, f_\ell,}_{n \text{ times}}$};
    \node(x5) [right = 0cm of x4] {$\underbrace{f_{\ell+1}, \ldots, f_{\ell+1},}_{n \text{ times}}$};
    \node(x6) [right = 0cm of x5] {$\underbrace{f_{\ell+2}, \ldots, f_{\ell+2},}_{n \text{ times}}$};
    \node(x7) [right = 0cm of x6] {$\ldots \ldots$};
    \node(x8) [right = 0cm of x7] {$\underbrace{f_{m}, \ldots, f_{m}}_{n \text{ times}}$};
    \draw[thick,|-|] (-1,1) -- (6,1);
    \node at (2.7, 1.5) {$I(F,\ell)$};
\end{tikzpicture}
\caption{$I(F)$ consists of $n$ customers from each fare arriving in increasing order.}
\label{fig:F_inst}
\end{figure}

Next, define $I(A, i)$ to be the instance $I(A)$ truncated at the fare $f_i$, including fare $f_i$. More precisely, $I(A, i) = \{s(A)_t\}_{t=1}^{T(i)}$, where $T(i) = i n$ for $i \leq \ell$ and $T(i) = \ell n + \sum_{j = \ell+1}^{i} A_j$ for $i > \ell$. These instances represent prefixes of $I(A)$ and aim to exploit the non-clairvoyance of online algorithms: If an online algorithm receives the partial instance $I(A, i)$, the complete instance can be $I(A, k)$ for any $k \geq i$.

To define the hard instances, we will also need the instances in which fares arrive in increasing order of fares, in blocks of $n$ customers of each fare type. Let $I(F) = \{s(F)_t\}_{t=1}^{mn}$ be the instance in which $n$ customers from each fare type arrive in increasing order, i.e., $s(A)_t = f_j$ if $(j-1)n < t \leq j n$ for some $j \in [m]$ (see Figure~\ref{fig:F_inst}). Furthermore, let $I(F,i) = \{s(F)_t\}_{t=1}^{i n}$ denote the instance obtained by pruning $I(F)$ to contain fares $f_i$ or lower.

\revision{Given two instances $I_1, I_2 \in \mathcal{U}$, we use $I_1 \oplus I_2$ to denote the instance in which the sequence of fares in $I_1$ is followed by the sequence of fares in $I_2$.} With these definitions in place, we are now ready to define the set of hard instances:
\begin{align*}
	\I = \{ I \in \mathcal{U} \mid I = I(A, k) \oplus I(F, i) \text{ for some } k,i \in [m]\} \cup \{I(A,k) \mid k \in [m]\}\,.
\end{align*}

Before proceeding to our LP-based upper bound that makes use of these hard instances $\I$, we pause to provide some intuition. Any online algorithm that is $\beta$-consistent needs to attain a revenue greater than or equal to $\beta \cdot Opt(I(A))$ when presented with the instance $I(A)$. But, given a prefix of $I(A)$, a non-clairvoyant algorithm has no way of knowing whether the full instance will be $I(A)$ or not. Therefore, it has to have the same accept/reject decisions on $I(A)$ and every prefix. Furthermore, to be $\gamma$-competitive, the online algorithm has to achieve $\gamma$-competitiveness on every instance. This places a constraint on the amount of capacity that an algorithm can commit in its pursuit of $\beta$-consistency because the instance can be any prefix of $I(A)$ (including $I(A)$ itself) followed by $I(F)$.

\subsection{LP-based Upper Bound}\label{subsec:upper-bound}

Next, we define a Linear Program (LP) that uses the hard instances $\I$ to bound $\beta(A, \gamma)$ from above. We use the variables $x_k$ to capture the number of customers of type $f_k$ accepted in $I(A)$ or any of its prefixes $\{I(A,k)\}_{k=1}^m$. Moreover, we use $y(k)_j$ to denote the number of customers of type $f_j$ accepted from the $I(F)$ part of $I(A,k) \oplus I(F)$. Finally, we use $N_i$ to denote the number of customers of fare type $f_i$ in the instance $I(A)$, i.e., $N_i = A_i$ if $i \geq \ell+1$ and $N_i = n$ otherwise. The following LP maximizes the level of consistency $\beta$ of any online algorithm that is $\gamma$-competitive against prefixes of $I(A)$ and all the hard instances. First, for feasibility, \eqref{constrant:capacity} enforces the capacity constraint for all hard instances $\{I(A, k) \oplus I(F)\}_{k=1}^m$. \eqref{constraint:prefix-comp} ensures $\gamma$-competitiveness on prefixes of $I(A)$, given by $\{I(A,k)\}_{k=1}^m$. \eqref{constraint:consistent} guarantees $\beta$-consistency on the advice instance $I(A)$. \eqref{constraint:comp} ensures $\gamma$-competitiveness on the hard instances $\{I(A,k) \oplus I(F)\}_{i,k}$.


\begin{varsubequations}{LP-1}\label{LP}
  \begin{align}
     & \max_{\beta,x,y} & \beta \nonumber\\
     & \text{s.t.} & n \geq \sum_{j=1}^k x_j + \sum_{j=1}^m y(k)_j                 && \forall k \in [m] && \label{constrant:capacity} \\
     &             & \frac{\sum_{j=1}^k f_j x_j}{Opt(I(A,k))} \geq \gamma                  && \forall k \in [m]        && \label{constraint:prefix-comp} \\
     &             & \frac{\sum_{j=1}^m f_j x_j}{Opt(A)} \geq \beta && \label{constraint:consistent} \\
     &             & \frac{\sum_{j=1}^k f_j x_j + \sum_{j=1}^i f_j y(k)_j}{Opt(I(A,k) \oplus I(F,i))} \geq \gamma && \forall i,k \in [m] \label{constraint:comp}\\
     &             & 0 \leq x_i \leq N_i && \forall i \in [m] \label{constraint:x-bound}\\
     &             & y(k)_j \geq 0 && \forall k,j \in [m]
  \end{align}
\end{varsubequations}


Consider an online algorithm $P$ that is $\gamma$-competitive and $\beta$-consistent on advice $A$. Let $X_j$ be the random variable that denotes the number of customers of fare type $f_j$ that $P$ accepts when presented with instance $I(A)$. Since $I(A, k)$ is a prefix of $I(A)$ for all $k \in [m]$, $X_j$ also denotes the number of customers of fare type $f_j$ that $P$ accepts when presented with the instance $I(A, k)$ for all $j \leq k$. Furthermore, for $j,k \in [m]$, let $Y(k)_j$ be the random variable that denotes the number of customers of fare type $f_j$ accepted by $P$ after time step $t = \sum_{i = 1}^m N_i$ in the instance $I(A,k) \oplus I(F)$. Since the instance $I(A,k) \oplus I(F,i)$ is a prefix of $I(A,k) \oplus I(F)$, for $j \leq i$, the random variable $Y(k)_j$ also denotes the number of customers with fare $f_j$ accepted by $P$ after time $t = \sum_{i = 1}^n N_i$ when presented with the instance $I(A,k) \oplus I(F,i)$.

Now, set $x_j = E[X_j]$ for all $j \in [m]$ and $y(k)_j = E[Y(k)_j]$ for all $j,k \in [m]$. Then, by virtue of linearity of expectation, we get that the feasibility, $\beta$-consistency and $\gamma$-competitiveness of $P$ implies the feasibility of $(\beta, (x_j)_j, (y(k)_j)_{j,k})$ as a solution of \ref{LP}, thereby implying $\beta \leq \beta^* $, where $\beta^*$ denotes the optimal value of \ref{LP}. Hence, we get the following upper bound:

\begin{theorem}\label{thm:upper-bound}
	For a given advice $A$ and a required level of competitiveness $\gamma \in [0, c(F)]$, let $\beta^*$ be the optimal value of \ref{LP}. Then, $\beta(A, \gamma) \leq \beta^*$.
\end{theorem}

In the next subsection, we use \ref{LP} to describe an online algorithm that is $\gamma$-competitive, and $\beta^*$-competitive on advice $A$. As a consequence, we will get that $\beta^* = \beta(A, \gamma)$. We conclude by noting that \ref{LP} has some redundancies in the constraints and the variables, e.g., \eqref{constraint:prefix-comp} for $k \leq \ell$ implies  \eqref{constraint:comp} for $i \leq k \leq \ell$, and $y(k)_j = 0$ for $j \leq k \leq \ell$ in every optimal solution. While the linear programming formulation can be tightened by removing the additional constraints and variables, we kept these to simplify the exposition.

\revision{
\subsection{LP-based Optimal Algorithm}\label{subsec:optimal-alg}

\begin{figure}[t!]
\begin{algorithm}[H]
   \caption{LP-based Optimal Algorithm}
   \label{alg:opt-alg}
    \begin{algorithmic}\vspace{0.08cm}
    		\item[\textbf{Input:}] Required level of competitiveness $\gamma \in [0,c(F)]$, advice $A$ and instance $I = \{s_t\}_{t=1}^T$.
    		\vspace{0.5em}
    		\item[$\mathbf{t = 0}$:] Solve \ref{LP} to find optimal solution $(\beta^*, x, (y(k))_k)$. Define protection levels:
    			\begin{itemize}
    				\item For $i \in [m]$, set $Q'_i = \sum_{j=1}^i x_j$. Moreover, set $Q_0' = 0$.
					\item For $i,k \in [m]$, set $R(k)_i = Q'_i + \sum_{j=1}^i y(k)_j$ if $i \leq k$ and $R(k)_i = Q'_k+ \sum_{j=1}^i y(k)_j$ if $i > k$. Moreover, set $R(k)_0 = 0$ for all $k \in [m]$.
    			\end{itemize}
    		\item[\textbf{Initialize:}] Total accepted customers with fare $f_j$ or below as $q_j = 0$ for all $j \in [m]$; total rejected customers with fare $f_j$ as $r_j = 0$ for all $j \in [m]$; trigger $\delta = 0$; and active protection levels $Q^*_i = Q'_i$ for all $i \in [m]$.
    		\vspace{0.5em}
            \item[\textbf{For}] $t=1$ to $T$:
            \begin{enumerate}
                \item Let $p \in [m]$ be the fare class of customer $s_t$, i.e., $s_t = f_p$. Calculate the fraction of customer $s_t$ that would be accepted under $Q^*$: set $w^* = \max\{w \in [0,1] \mid q_j + w \leq Q_j^*\ \forall j \geq p\}$ and let $h \geq p$ be the largest index such that $q_p + w^* = Q^*_j$ (set $h= 0$ if no such index exists). If $\delta = 0$, update $r_p \leftarrow r_p + (1 - w^*)$ to reflect the anticipated rejection of a $1 - w^*$ fraction of $s_t$.

                \item \emph{Check Trigger Condition.} If $\delta = 0$ and $\sum_{i=j}^h r_i > \sum_{i=j}^h \{A_i - x_i\}$ for some $j \geq \ell + 1$, i.e., the total rejections among fares $\{f_j, \dots, f_h \}$ exceeds the rejections among fares $\{f_j, \dots, f_m\}$ in $I(A)$ under LP solution $(\beta^*, x, (y(k))_k)$:
                	\begin{itemize}
                		\item Set trigger $\delta = 1$ and $k = \max\{j \mid q_j = Q_j'\}$ (set $k = 1$ if no such index exists).
                		\item \textbf{While} $\exists\ j \in [m]$ such that $R(k)_j < q_j$: Set $k = \min\{j \in [m] \mid R(k)_j < q_j\}$.\\
                			Set $k^* = k$, i.e., $k^*$ is the value with which the `While' loop terminates.
                		\item \emph{Switch Protection Levels.} Set $Q^*_j = R(k^*)_j$ for all $j \in [m]$.
                	\end{itemize}

                \item \emph{Make decision according to $Q^*$.} Accept $w^* = \max\{w \in [0,1] \mid q_j + w \leq Q_j^*\ \forall j \geq p\}$
                 fraction of customer $s_t$ and update $q_j \leftarrow q_j + w$, for all $j \geq p$, to reflect the increase in number of customers of fare $f_j$ or lower accepted by the algorithm.

            \end{enumerate}
    \end{algorithmic}
 \end{algorithm}
\end{figure}

In this subsection, we describe an algorithm (Algorithm~\ref{alg:opt-alg}) that is $\beta^*$-consistent for advice $A$ while also being $\gamma$-competitive. This allows us to establish the main result of this section:

\begin{theorem}\label{thm:lower-bound}
	For any advice $A$ and a required level of competitiveness $\gamma \in [0, c(F)]$, the optimal value of \ref{LP}, $\beta^*$, is equal to $\beta(A, \gamma)$. Furthermore, Algorithm~\ref{alg:opt-alg} is $\beta(A, \gamma)$-consistent on advice $A$ and $\gamma$-competitive.
\end{theorem}

Note that, in particular, Theorem~\ref{thm:lower-bound} implies that \ref{LP} is an efficient method of computing the consistency-competitiveness Pareto frontier. Furthermore, for every advice $A$, Algorithm~\ref{alg:opt-alg} provides a way to achieve levels of consistency and competitiveness that lie on this Pareto frontier, thereby characterizing the tradeoff between consistency and competitiveness for every advice. In the rest of this subsection, we motivate Algorithm~\ref{alg:opt-alg} and sketch the proof of Theorem~\ref{thm:lower-bound}. The full proof is relegated to Appendix~\ref{appendix:opt-alg}.

Algorithm~\ref{alg:opt-alg} is based on protection level policies and operates in two phases. Throughout its run, it maintains active protection levels $Q^*$, and makes accept/reject decisions based on them. Importantly, these active protection levels are not fixed and take on different values in the two phases. The transition between the two phases in captured by the trigger $\delta$. Since the protection levels change with time, Algorithm~\ref{alg:opt-alg} is \emph{not} a protection level policy. The protection levels of both phases are determined using an optimal solution $(\beta^*, x, \{y(k)\}_{k=1}^m)$ of \ref{LP}. In particular, we use the solution to construct protection levels, denoted by $Q'$ for phase $\delta = 0$ and $m$ candidate protection levels $\{R(k)\}_{k=1}^m$ for phase $\delta = 1$, defined as follows:
\begin{align*}
	Q_i' = \sum_{j=1}^i x_j \quad \text{ and } \quad R(k)_i = \sum_{j=1}^{\min\{i,k\}} x_j + \sum_{j=1}^i y(k)_i = \begin{cases}
		Q_i' + \sum_{j=1}^i y(k)_j & \text{if } i \leq k\\
		Q_k' + \sum_{j=1}^i y(k)_j & \text{if } i > k
	\end{cases}\,.
\end{align*}
These protection levels are motivated by the hard instances. In particular, $Q'$ is devised to achieve a reward of $\sum_{j=1}^i x_j \cdot f_j$ on instance $I(A, i)$ and $R(k)$ is devised to achieve a revenue of $\sum_{j=1}^k x_j \cdot f_j + \sum_{j=1}^i y(k)_j \cdot f_j$ on instance $I(A,k) \oplus I(F, i)$.

Initially, the active protection levels $Q^*$ are set equal to $Q'$ with the goal of achieving good performance on $I(A)$, or more generally for any instance that conforms to the advice. If the advice is accurate, we show that the algorithm accrues a reward larger than $\sum_{j=1}^m f_j x_j$, which in turn is larger than $\beta^* \cdot I(A)$ because of \eqref{constraint:consistent}. But this is not enough as the algorithm also needs to be $\gamma$-competitive. If the algorithm keeps using $Q'$ to make decisions even when the advice is inaccurate, it risks rejecting too many customers and violating the $\gamma$-competitiveness requirement. Therefore, it keeps track of the loss from rejecting customers in the form of $\{r_j\}_{j=1}^m$, and uses this information to trigger the phase transition.

Intuitively, the phase transition should be triggered if continuing with protection levels $Q'$ would result in a violation of the $\gamma$-competitiveness requirement. To that end, in Step 1 we calculate the additional rejection $(1 -w^*)$ that would result from the continued use of $Q'$ and update the rejections $\{r_j\}_{j=1}^m$. In order to ensure that these rejections do not violate $\gamma$-competitiveness, the algorithm compares them to rejections made by \ref{LP} on the instance $I(A)$: it allocates a rejection-capacity of $A_i - x_i$ to fare $i$ for all $i \geq \ell+1$. Here, fare classes $i \leq \ell$ are ignored because we do not require the advice to be good on fare classes $i \leq \ell$ which are not predicted to form a part of the optimal solution. To generate even greater leeway for the instance to deviate from the advice, the algorithm enforces these capacities in an upward-closed manner, i.e., we allow rejections of lower fares to cannibalize the rejection-capacity of higher fares (up to fare $f_h$). In particular, the trigger condition checks whether $\sum_{i=j}^h r_i > \sum_{i=j}^h \{A_i - x_i\}$ for some $j \geq \ell + 1$. The summation is limited to $h$ because the algorithm has filled the protection levels till fare $f_h$, i.e., $q_h = Q_h$, and we can charge the rejections against these acceptances to ensure $\gamma$-competitiveness. The definition of $h$ implies that the algorithm accepts every customer with fare larger than $f_h$, however the number of such customers can potentially be zero and we cannot charge the rejections against them.

Once the trigger is invoked, the algorithm needs to select new protection levels which satisfy the following crucial properties: (i) the number of customers accepted till the trigger pull should not violate these new protection levels; and (ii) ensure $\gamma$-competitiveness. Algorithm~\ref{alg:opt-alg} achieves this by carefully choosing a set of protection levels $R(k^*)$ from the collection $\{R(k)\}_{k=1}^m$. It sets $Q^* = R(k^*)$ to be the new active protection levels and uses them to make decisions on the remaining customers.

The proof of Theorem~\ref{thm:lower-bound} involves a fine-grained analysis of the performance of Algorithm~\ref{alg:opt-alg}. We begin by showing that Algorithm~\ref{alg:opt-alg} achieves a $\beta$ fraction of $Opt(I)$ for any instance $I \in S(A)$ that conforms to the advice. Note that no assumptions are made about the order in which the fares arrive. Since the order of the fares can effect when the phase transition occurs, we cannot directly leverage the proof techniques developed for analyzing protection level policies. Consequently, we first carefully analyze the phase transition rule to ensure that it is not triggered on any instances $I \in S(A)$ which conforms to the advice and then bound the performance of  the protection level policy based on $Q'$ by leveraging \ref{LP}.

Consistency is not the only requirement and we also need to ensure $\gamma$-competitiveness of Algorithm~\ref{alg:opt-alg}. As mentioned earlier, Algorithm~\ref{alg:opt-alg} is not a protection level policy, and consequently cannot be analyzed as one. In light of the lack of structure imposed on the instance and the novel phase transition that is part of Algorithm~\ref{alg:opt-alg}, apriori it is not even clear what the hard instances for Algorithm~\ref{alg:opt-alg} look like, and even though we referred to the instances defined in Subsection~\ref{subsec:hard-instances} as `hard', it is not clear that $\gamma$-competitiveness on hard instances of Subsection~\ref{subsec:hard-instances} ensures $\gamma$-competitiveness on all instances. The design of the hard instances goes hand-in-hand with the proof of Theorem~\ref{thm:lower-bound}, with the latter guiding the former. Establishing $\gamma$-competitiveness is rife with technical challenges: (i) the trigger condition for the phase transition relies on rejections, compelling us to keep track of rejections, which is not required for analyzing protection level policies; (ii) ensuring that using $Q'$ (which is designed to perform well when the advice is true) to make decisions before the phase transition does not cause Algorithm~\ref{alg:opt-alg} to make irrevocable decisions that preclude $\gamma$-competitiveness. We tackle (i) by developing finer bounds on $Opt(I)$ which take rejections into account and connecting them to the reward accumulated by Algorithm~\ref{alg:opt-alg}, and establish (ii) by performing an intricate case-analysis of the procedure used by Algorithm~\ref{alg:opt-alg} to select the post-transition protection levels. Our analysis has to continually contend with the lack of order in the instance and its effect on the phase transition, making (i) and (ii) even more challenging.
}

\section{Optimal Protection Level Policy}\label{sec:protec}

By carefully constructing and switching between protection levels,  Algorithm~\ref{alg:opt-alg} employs protection level policies with different protection levels as sub-routines to achieve performance that lies on the consistency-competitiveness Pareto frontier. In particular, it runs the protection level policy based on $Q'$ when the trigger is $\delta = 0$ and runs the protection level policy based on $R(k^*)$ once $\delta = 1$. This switch between protection levels can lead to non-monotonicity: Algorithm~\ref{alg:opt-alg} may accept a customer of a particular fare type having rejected another customer of that same fare type at an earlier time step (see Example~\ref{example:bad}). Non-monotonicity ``can lead to the emergence of strategic consumers or third parties that specialize in exploiting inter-temporal fare arbitrage opportunities, where one waits for a lower fare class to be available. To avoid such strategic behavior, the capacity provider may commit to a policy of never opening fares once they are closed'' \citep{gallego2019revenue}. Additionally, in order to switch between protection levels, Algorithm~\ref{alg:opt-alg} also requires information about the fare types of customers that were rejected, which may not be available depending on the context.

In contrast, using a single protection level policy on the entire instance (i.e., using the same protection levels at all time steps) avoids these issues, in addition to satisfying a multitude of other practically-relevant properties which make protection level policies the bedrock of single-leg revenue management (see~\citealt{talluri2004theory} for a detailed discussion). As we will show later in this section, this insistence on using a single set of protection levels at all time steps comes with a loss of optimality. Nonetheless, the practical importance of protection level policies warrants an investigation into the optimal way to incorporate advice for protection level policies. To that end, in this section, we characterize the protection level policy that attains the optimal level of consistency for a given advice $A$ among all protection level policies that are $\gamma$-competitive.

\begin{definition}
	Fix an advice instance $A$ and a level of competitiveness $\gamma \in [0, c(F)]$. Let $\mathcal{P}^{PL}_\gamma$ denote the set of all protection level policies that are $\gamma$-competitive. Then, the optimal level of consistency for advice $A$ that can be attained by protection level policies under the $\gamma$-competitiveness requirement is defined as follows:
\begin{align*}
	\beta^{PL}(A, \gamma) \coloneqq \max_{P \in \mathcal{P}^{PL}_\gamma} \min_{\{s_t\}_{t=1}^T \in S(A)} \frac{\mathbb{E}_P\left[\sum_{t=1}^T P_t(A, \{s_i\}_{i=1}^{t-1}, s_t) \cdot s_t\right]}{Opt(A)}\,.
\end{align*}
\end{definition}


\revision{
\subsection{Computing the Optimal Protection Levels}\label{subsec:opt-protec}

\begin{figure}[t]
\begin{algorithm}[H]
   \caption{Optimal Protection Levels}
   \label{alg:opt-protec}
    \begin{algorithmic}\vspace{0.08cm}
    		\item[\textbf{Input:}] Required level of competitiveness $\gamma \in [0, c(F)]$, advice $A$,  and approximation parameter $\epsilon > 0$.
    		\item[\textbf{function}] CoreSubroutine$(\beta)$: \Comment{Core subroutine}\\
    			\begin{enumerate}
    				\item Set $Q_i = 0$ for $0 \leq i \leq m$.\label{alg:prot levels beta}
            		\item \textbf{For} $k = 1$ to $m$:\label{alg:prot levels for}
            			\begin{enumerate}
            				\item Define
            					\begin{align}\label{eqn:opt-protec-ck}
            						c_k \coloneqq  \frac{\gamma \cdot Opt(I(F,k)) - Q(I(F, k-1))}{f_k}\,,
            					\end{align}
            					and set $Q_i \leftarrow Q_{k-1} + c_k$ for all $i \geq k$.
            				\item If $Q(I(A,k)) + \sum_{i=k+1}^m N_i f_i < \beta \cdot Opt(I(A))$: Define
            					\begin{align}\label{eqn:opt-protec-dk}
            						d_k \coloneqq \frac{\beta \cdot Opt(I(A)) - Q(I(A, k)) - \sum_{i = k+1}^m N_i f_i}{f_k}\,,
            					\end{align}
                      then set $Q_i \leftarrow Q_{i} + d_k$ for all $i \geq k$.
            			\end{enumerate}
            		\item Return protection levels $Q$.
    			\end{enumerate}
    		\item[\textbf{end function}]
    		\item[\textbf{Initialize:}] $\underline \beta = c(F)$ and $\bar \beta = 1$. \Comment{Binary Search}
    		\item[\textbf{While}] $\bar \beta - \underline \beta > \epsilon$:
            	\begin{enumerate}
            		\item Set $\beta = (\bar \beta + \underline \beta)/2$.
            		\item $Q \leftarrow$ CoreSubroutine$(\beta)$.
            		\item If $Q_m > n$, set $\bar \beta = \beta$. Else if $Q_m \leq n$, set $\underline \beta = \beta$.
            	\end{enumerate}
            \item[\textbf{Return:}] $Q \leftarrow$ CoreSubroutine$(\underline \beta)$.
    \end{algorithmic}
 \end{algorithm}
\end{figure}

The goal of this subsection is to describe an algorithm (Algorithm~\ref{alg:opt-protec}) which, given advice $A$ and a required level of competitiveness $\gamma \in [0, c(F)]$, outputs protection levels $Q$ such that the protection level policy based on $Q$ is (nearly) $\beta^{PL}(A, \gamma)$-consistent on advice $A$. Before proceeding to the algorithm, we define a few terms necessary for its description. For the remainder of this subsection, fix an advice $A$ and a required level of competitiveness $\gamma \in [0, c(F)]$. Recall that protection levels are simply vectors $Q \in \R_+^m$ such that $Q_1 \leq Q_2 \leq \dots Q_m \leq n$. If a vector $Q \in \R_+^m$ satisfies $Q_1 \leq Q_2 \leq \dots Q_m$ but $Q_m > n$ then we call it infeasible. Given an instance $I$ and protection levels $Q$, we use $Q(I)$ to denote the revenue obtained by the protection level policy with protection levels $Q$ when acting on the instance $I$. Furthermore, recall that $N_i$ denotes the number of customers of fare type $f_i$ in the instance $I(A)$, i.e., $N_i = A_i$ if $i \geq \ell+1$ and $N_i = n$ otherwise. Finally, our description of Algorithm~\ref{alg:opt-protec} also makes use of the instances defined in subsection~\ref{subsec:hard-instances}.

The core subroutine in Algorithm~\ref{alg:opt-protec} takes as input a desired level of consistency $\beta$ and decides whether any protection level policy can achieve $\beta$-consistency on advice $A$ while also being $\gamma$-competitive. It does so by iterating through the fares in increasing order $k = 1$ to $m$ and setting $Q_k$ to be smallest level that satisfies the following properties:
\begin{itemize}
	\item The (possibly infeasible) protection levels $Q_1 \leq \dots \leq Q_k$ ensure a revenue greater than or equal to $\gamma \cdot Opt(I(F,k))$ on instance $I(F,k)$. We define $c_k$ in \eqref{eqn:opt-protec-ck} to be such that $Q(I(F,k-1)) + c_k \cdot f_k = \gamma \cdot Opt(I(F,k))$, and then set $Q_k = Q_{k-1} + c_k$ to ensure $\gamma$-competitiveness on $I(F,k)$: $$Q(I(F,k)) = Q(I(F,k-1)) + (Q_k - Q_{k-1}) \cdot f_k = Q(I(F,k-1)) + c_k \cdot f_k = \gamma \cdot Opt(I(F,k))\,.$$
	\item The (possibly infeasible) protection levels $Q_1 \leq \dots \leq Q_k$ preserve the potential to attain a revenue greater or equal to $\beta \cdot Opt(I(A))$ on instance $I(A)$, i.e., if we were to accept all of the customers in $I(A)$ with fare strictly greater than $f_k$ (leading to a revenue of $\sum_{i = k+1}^m N_i f_i$), in addition to the fares accepted according to the protection level policy $Q$ till fare $f_{k}$ (which yield a revenue of $Q(I(A,k))$), then we can guarantee a revenue greater than or equal to $\beta \cdot I(A)$.  We define $d_k$ in \eqref{eqn:opt-protec-dk} to be such that $Q(I(A,k)) + d_k \cdot f_k + \sum_{i = k+1}^m N_i f_i = \beta \cdot Opt(I(A))$, and increase $Q_k$ by $d_k$ to preserve the potential to attain $\beta$-consistency:
		\begin{align*}
			\text{Updated-}Q(I(A,k)) + \sum_{i = k+1}^m N_i f_i = Q(I(A,k)) + d_k \cdot f_k + \sum_{i = k+1}^m N_i f_i = \beta \cdot Opt(I(A))\,.
		\end{align*}
\end{itemize}

In the proof of Theorem~\ref{thm:opt-protec}, we show that these protection levels $Q$ are feasible ($Q_m \leq n$) if and only if some feasible protection level policy can attain $\beta$-consistency on advice $A$ while being $\gamma$-competitive. To do so, we prove that the core subroutine minimizes the total capacity needed to achieve the required level of consistency and competitiveness, thereby implying that if any feasible protection level policy can achieve it, then the protection levels $Q$ are feasible. Algorithm~\ref{alg:opt-protec} simply uses this core subroutine to perform a binary search over the space of all consistency levels $\beta$ to find the largest one which can be attained for advice $A$ while being $\gamma$-competitive. The `Return' step ensures that Algorithm~\ref{alg:opt-protec} returns feasible protection levels: It sets $\beta = \underline \beta$ and computes the corresponding protection levels $Q$, which are guaranteed to satisfy $Q_m \leq n$ because of the binary search procedure.

\begin{theorem}\label{thm:opt-protec}
	If $Q$ are the protection levels returned by Algorithm~\ref{alg:opt-protec}, upon receiving required level of competitiveness $\gamma \in [0, c(F)]$, advice $A$, and approximation parameter $\epsilon > 0$ as input, then the protection level policy based on $Q$ is $(\beta^{PL}(A, \gamma) - \epsilon)$-consistent on advice $A$ and $\gamma$-competitive. Moreover, the runtime of Algorithm~\ref{alg:opt-protec} grows polynomially with the input size and $\log(1/\epsilon)$.
\end{theorem}

The proof of Theorem~\ref{thm:opt-protec} can be found in Appendix~\ref{appendix:protec}. The technical challenge lies in establishing the following fact: If \emph{any} protection level policy can achieve $\beta'$-consistency and $\gamma$-competitiveness, then the core subroutine with $\beta = \beta'$ returns protection levels $Q$ which are feasible (i.e., $Q_m \leq n$). To achieve this, we iteratively define linear programs, each of which builds on the solution of previous programs, in order to connect the consistency and competitiveness of an arbitrary protection level policy to that of Algorithm~\ref{alg:opt-protec}.
}

%
%
%

\subsection{Sub-optimality of Protection Level Policies}\label{subsec:protec-level-suboptimal}

As we alluded to at the beginning of this section, protection levels policies do not always yield performance that lies on the consistency-competitiveness Pareto frontier. In fact, as we show in this subsection, in some cases it can be quite far from the Pareto frontier. We do so by exhibiting an example that illustrates the sub-optimality of protection level policies.

\begin{exmp}\label{example:bad}
	Assume that the capacity $n$ satisfies $n\ \text{(mod)}\ 9 = 0$. Set $f_1 = 1 < f_2 = \eta < f_3 = \eta^2$, where $\eta \gg n$ is a large number. Moreover, set $A = (1, n/3, \{2n/3\} - 1)$, and $\gamma = 1/3$. Consider $\beta = (1/2) + (3/n)$. Then, using the notation of Algorithm~\ref{alg:opt-protec}, we get that $c_1 = n/3$, $d_1 = 0$, $c_2 = (1 - 1/\eta)(n/3)$, $d_2 = 0$, $c_3 = (1 - 1/\eta)(n/3)$ and $d_3 \geq \frac{\beta \cdot [\eta^2 (2n/3 - 1)] - (\eta^2 n/3)}{\eta^2} = \frac{(\eta^2 n/3) + (2 - \beta)\eta^2 -  (\eta^2 n/3)}{\eta^2} = 2 - \beta\,.$
	
	Therefore, $Q_3 = c_1 + c_2 + c_3 + d_1 + d_2 + d_3 > n$. Hence, using Lemma~\ref{lemma:protec-second} in Appendix~\ref{sec:appendix:protec}, we get that $\beta^{PL}(A, \gamma) \leq (1/2) + (6/\eta) \approx 1/2$. Next, note that the following values form a feasible solution of \ref{LP}: $\beta = 5/(6 + n/\eta) \approx 5/6$ and
	\begin{align*}
		\quad &x_1 = n/3; \quad x_2 = n/9; \quad x_3 = 5n/9;\\
		&y(1)_1 = 0; \quad y(1)_2 = n/3; \quad y(1)_3 = n/3;\\
		&y(2)_1 = 0; \quad y(2)_2 = 2n/9; \quad y(2)_3 = n/3;\\
		&y(3)_1 = 0; \quad y(3)_2 = 0 ; \quad y(3)_3 = 0\,.
	\end{align*}
	Hence, we get that $\beta(A, \gamma) \geq 5/(6 + n/\eta)$, thereby showing that, in the worst case, protection level policies may be far from the optimal. This example also demonstrates the lack of monotonicity in Algorithm~\ref{alg:opt-alg}. The solution described above is an optimal integral solution of the LP and very close to the optimal solution of the LP. If Algorithm~\ref{alg:opt-alg} receives the instance $I(A, 2) \oplus I(F)$, it will initially reject some customers with fare $f_2$ because $x_2 = n/9$ and $A_2 = n/3$, and then accept customers with fare $f_2$ when they arrive as part of $I(F)$ because $y_2 = 2n/9$.
\end{exmp}

%
%

We conclude this section by providing some intuition behind the sub-optimality of protection level policies that motivated the above example. The central difference between the LP-based optimal algorithm (Algorithm~\ref{alg:opt-alg}) and the optimal protection level policy (Algorithm~\ref{alg:opt-protec}) is that the former changes the protection levels it uses if the instance deviates from the advice. Hence, if the advice is turning out to be true then the optimal LP-based algorithm can use the minimum amount of capacity needed to maintain $\gamma$-competitiveness on the prefix of the instance seen so far, in an attempt to conserve capacity and deploy it on customers with higher fares that will arrive later, thereby attaining high consistency. To do so, it uses the protection levels $Q'$, as defined in Algorithm~\ref{alg:opt-alg}, while the advice is true. Crucially, $Q'$ need not be $\gamma$-competitive on instances that cannot go on to conform with the advice, and is optimized for consistency. This does not violate the competitiveness requirement because if the instance deviates from the advice, the LP-based optimal algorithm stops using $Q'$ and switches to protection levels that allow it to attain $\gamma$-competitiveness. On the other hand, the optimal protection level policy lacks this adaptivity and uses the same protection levels throughout. This forces it to use more capacity on lower fares because it needs to maintain $\gamma$-competitiveness on instances in which a large number of such customers arrive, even though the advice predicts a low number and the instance is consistent with the advice. This ``overcommitment'' to lower fares prevents it from attaining the optimal level of consistency. In our empirical evaluation of the two algorithms (see Section~\ref{sec:simulation}), we found this overcommitment to be the main driver of the sub-optimality of protection level policies: In settings that are not geared to exploit this overcommitment, the optimal protection level policy (Algorithm~\ref{alg:opt-protec}) attains performance on the consistency-competitiveness Pareto frontier.

\revision{
\section{Robustness}\label{sec:robustness}

Till now, the performance of an algorithm has been evaluated along two dimensions: (i) consistency considers the performance of an online algorithm when the instance $I$ conforms to the advice ($I \in S(A)$); and (ii) competitiveness considers the performance of an online algorithm on all instances. In practice, we would also like our online algorithm to perform well in the vicinity of the advice $A$. To that end, in this section, we show that the performance of protection level policies degrades gracefully as the instance deviates from the advice; and a natural modification of Algorithm~\ref{alg:opt-alg} also exhibits graceful performance degradation when the instance is close to conforming with the advice. We begin by defining a notion of dissimilarity (or distance) between the instance and the advice, which will allow us to quantify the robustness to deviations. The distance $\lambda(I, A)$ between an advice $A$ and an instance $I = \{s_t\}_{t=1}^T$ is given by
\begin{align}\label{eqn:distance}
	\lambda(I,A) \coloneqq \left(A_\ell - K_I(\ell) \right)^+ + \sum_{j=\ell+1}^m \left|A_j - K_I(j) \right| \,,
\end{align}
where $K_I(j) = \sum_{t=1}^T \mathds{1}(s_t = f_j)$ represents the number of customers of fare type $f_j$ in instance $I$, and, as before, $\ell = \min\{ i \in [m] \mid A_i \geq 1\}$ is the lowest fare that forms a part of $A$. The definition of $\lambda$ is motivated by the following equivalence: $\lambda(I, A) = 0$ \emph{if and only if} $I\in S(A)$, i.e., $I$ conforms to advice $A$. Recall that we use $Q(I)$ to denote the total revenue collected by the protection level policy with protection levels $Q = (Q_1, \dots, Q_m)$. The following result shows that the performance of protection level policies degrades smoothly as a function of the distance between the instance and the advice.

\begin{theorem}\label{thm:protec-robust}
	For every advice $A$, instance $I$ and protection levels $Q = (Q_1, \dots, Q_m)$, we have
	\begin{align*}
		  Q(I) \geq \frac{Q(I(A))}{Opt(A)} \cdot Opt(I) - 2f_m \cdot \lambda(I,A)\,.
	\end{align*}
\end{theorem}

We can leverage Theorem~\ref{thm:protec-robust} to show that the optimal protection levels returned by Algorithm~\ref{alg:opt-protec} are robust to deviations in the advice.

\begin{corollary}
	Let $Q = (Q_1, \dots, Q_m)$ be the protection levels returned by Algorithm~\ref{alg:opt-protec} when given required level of competitiveness $\gamma$, advice $A$ and approximation parameter $\epsilon > 0$ as input. Then, for every instance $I$, we have $Q(I) \geq \left(\beta^{PL}(A, \gamma) - \epsilon \right) \cdot Opt(I) - 2f_m \cdot \lambda(I,A)$.
\end{corollary}

Next, for $\Lambda>0$, we define a variation of Algorithm~\ref{alg:opt-alg} called the $\Lambda$-relaxed LP-based Optimal Algorithm through the following changes:
\begin{itemize}
	\item Set $Q'_i = \sum_{j=1}^i \lfloor x_j \rfloor$ instead of $Q'_i = \sum_{j=1}^i x_j$.
	\item Replace the trigger condition `If $\delta = 0$ and $\sum_{i=j}^h r_i > \sum_{i=j}^h \{A_i - x_i\}$ for some $j \geq \ell + 1$' in Step 2 with `If $\delta = 0$ and $\sum_{i=j}^h r_i > \Lambda + \sum_{i=j}^h \{A_i - x_i\}$ for some $j \geq \ell + 1$'.
\end{itemize}
The rest of the algorithm is identical to Algorithm~\ref{alg:opt-alg} (The precise algorithm can be found in Appendix~\ref{appendix:robustness}). Recall that Algorithm~\ref{alg:opt-alg} keeps track of the rejections, and changes the trigger from $\delta = 0$ to $\delta = 1$ when the rejection-capacity is exceeded. The intuition behind it being that if the deviation of the instance from the advice large enough to exceed the rejection capacity, then the advice is not accurate and we no longer need to focus on ensuring consistency. But, in the worst-case, this does not give the machine-learning algorithm which generates the advice much room for error.

 The $\Lambda$-relaxed Algorithm~\ref{alg:opt-alg} resolves this issue by increasing the rejection capacity by $\Lambda$. As long as the advice was good enough for the instance to satisfy $\lambda(I,A) \leq \Lambda$, the $\Lambda$-relaxed Algorithm~\ref{alg:opt-alg} would continue to aim for consistency. The next theorem establishes that this makes the $\Lambda$-relaxed Algorithm~\ref{alg:opt-alg} robust to deviations from the advice.

\begin{theorem} \label{thm:robust-opt}
	Consider some advice $A$ and a competitiveness level $\gamma \in [0, c(F)]$. Then, for $\Lambda \geq 1$:
	\begin{enumerate}
		\item The $\Lambda$-relaxed Algorithm~\ref{alg:opt-alg} attains a revenue greater than or equal to $\beta(A, \gamma) \cdot Opt(I) - 2 f_m \cdot \lambda(I,A)$ on any instance $I$ with $\lambda(I,A) \leq \Lambda$.
		\item The $\Lambda$-relaxed Algorithm~\ref{alg:opt-alg} attains a revenue greater than or equal to $\gamma \cdot Opt(I) - f_m \cdot (\Lambda + m)$ on any instance $I \in \mathcal{U}$.
	\end{enumerate}
\end{theorem}
}

\section{Numerical Experiments}\label{sec:simulation}

In this section, we numerically evaluate the performance of the LP-based optimal algorithm (Algorithm~\ref{alg:opt-alg}) and the optimal protection level policy (Algorithm~\ref{alg:opt-protec}) on a variety of advice and problem parameters. We do so with two central goals in mind: (i) comparing the performance of the two algorithms; (ii) evaluating the robustness of both algorithms, by computing the expected performance of both algorithms when the instance is drawn from a distribution centered on the advice. In this section, we will use the consistency of the $c(F)$-competitive protection level policy of \citet{ball2009toward} as a benchmark, which is equal to $Q(I(A))/Opt(A)$, where $Q_i = \left[\sum_{j=1}^i \left(1 - f_{i-1}/{f_i}\right)\right]\cdot c(F) \cdot n$ are the protection levels used by the policy of \citet{ball2009toward}.
 \revision{Note that the consistency-competitiveness frontier of the policy of \citet{ball2009toward} is always a point because the algorithm has a fixed competitiveness.} However, in general we do have that its consistency satisfies $Q(I(A)) / Opt(A) > c(F)$ because the algorithm of \citet{ball2009toward} performs better on instances that are well behaved. Moreover, consistency of their algorithm always lies entirely below \emph{the lowest point} of the consistency-competitiveness frontier of the optimal protection level policy (Algorithm~\ref{alg:opt-protec}) by virtue of our algorithm being optimal. In other words, for every advice $A$, the consistency of the policy of \citet{ball2009toward} is always less than $\beta^{PL}(A, c(F)) = \min_{\gamma \in [0, c(F)]} \beta^{PL}(A, \gamma)$. Based on this observation, it comes as no surprise that the LP-based optimal algorithm (Algorithm~\ref{alg:opt-alg}) and the optimal protection level policy (Algorithm~\ref{alg:opt-protec}) have much higher consistency than the policy of \citet{ball2009toward}, which is oblivious to the advice. Throughout this section, we fix the capacity to be $n = 100$.

\subsection{Comparing the Policies}

\begin{figure}[t!]
	\begin{subfigure}{0.33\textwidth}
		\includegraphics[width = \linewidth]{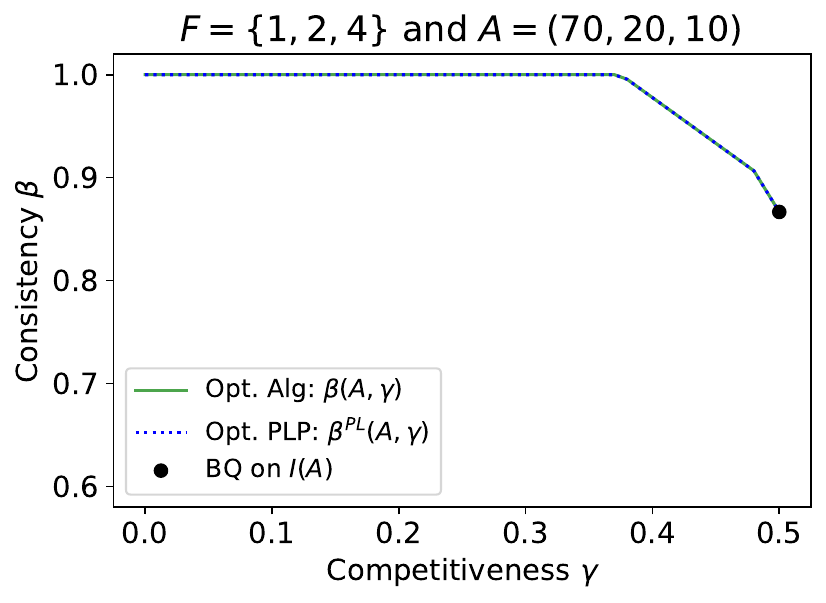}
	\end{subfigure}%
	\begin{subfigure}{0.33\textwidth}
		\includegraphics[width = \linewidth]{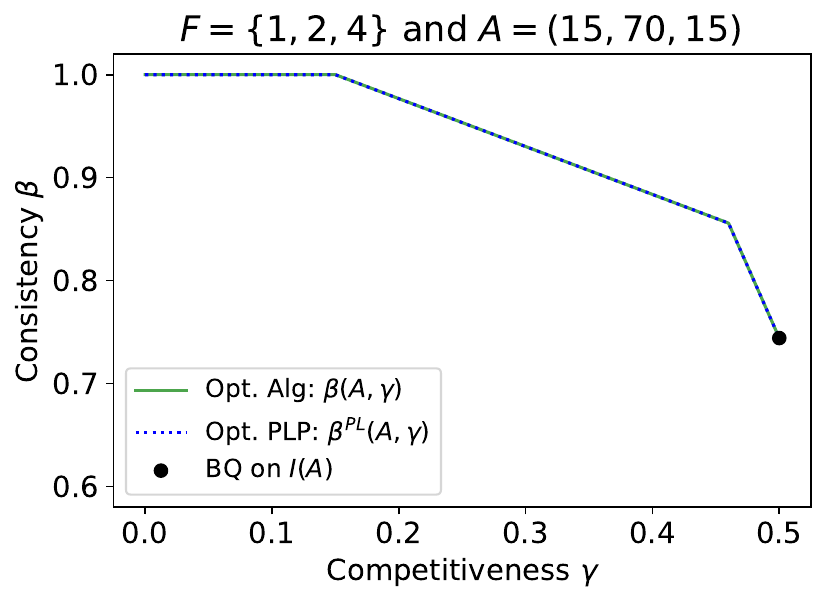}
	\end{subfigure}%
	\begin{subfigure}{0.33\textwidth}
		\includegraphics[width = \linewidth]{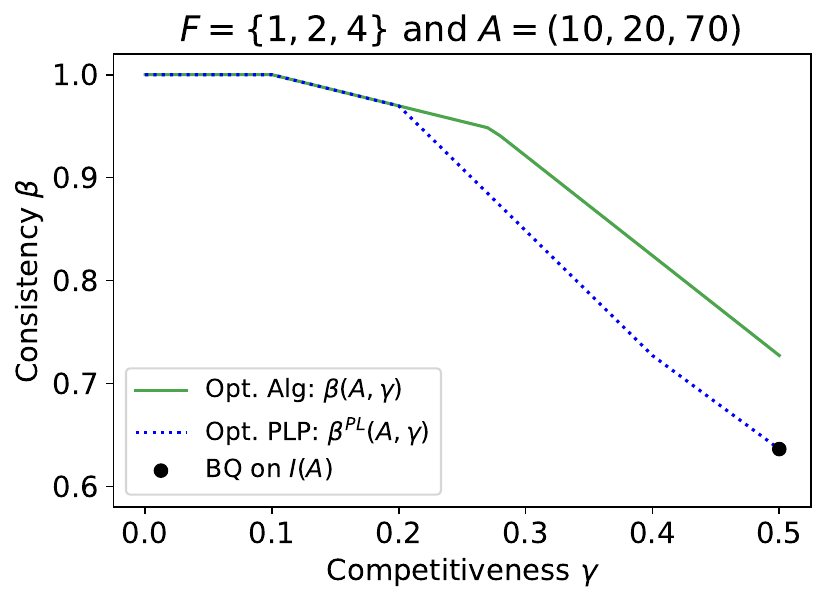}
	\end{subfigure}
	\begin{subfigure}{0.33\textwidth}
		\includegraphics[width = \linewidth]{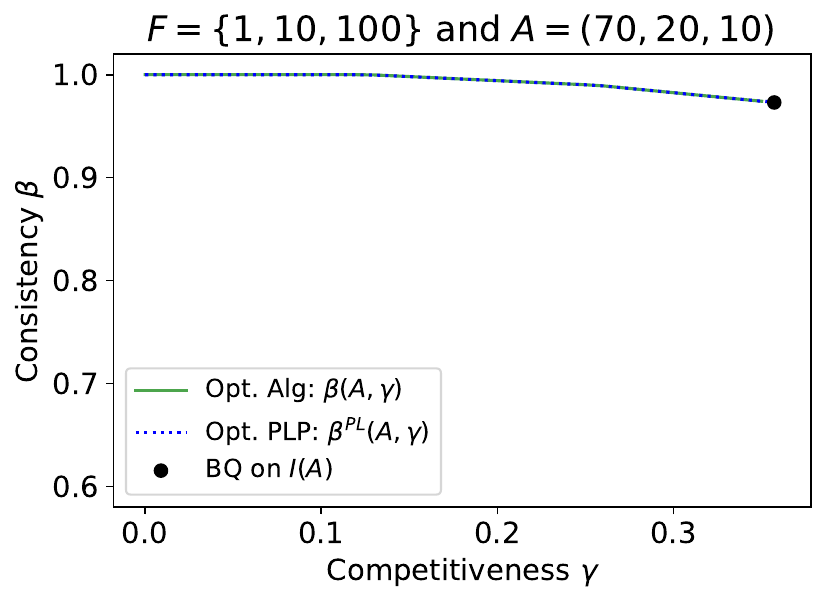}
	\end{subfigure}%
	\begin{subfigure}{0.33\textwidth}
		\includegraphics[width = \linewidth]{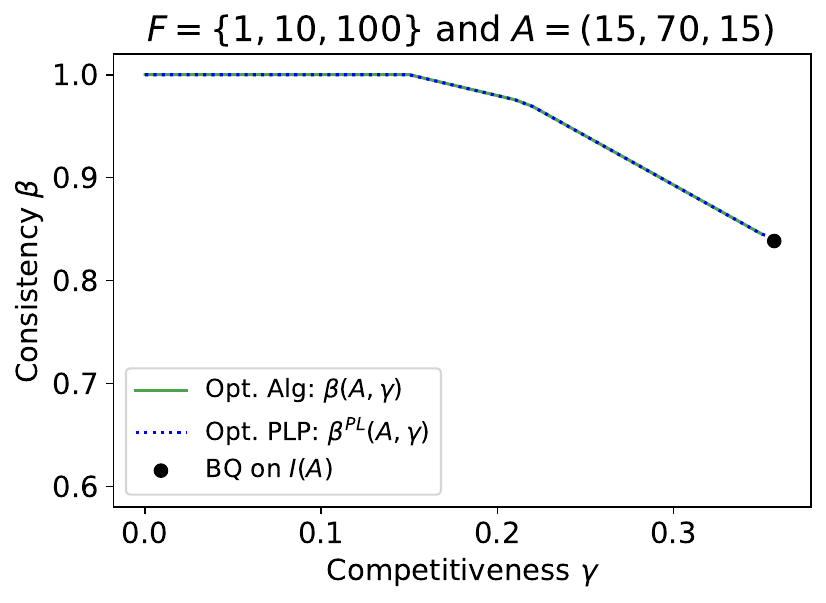}
	\end{subfigure}%
	\begin{subfigure}{0.33\textwidth}
		\includegraphics[width = \linewidth]{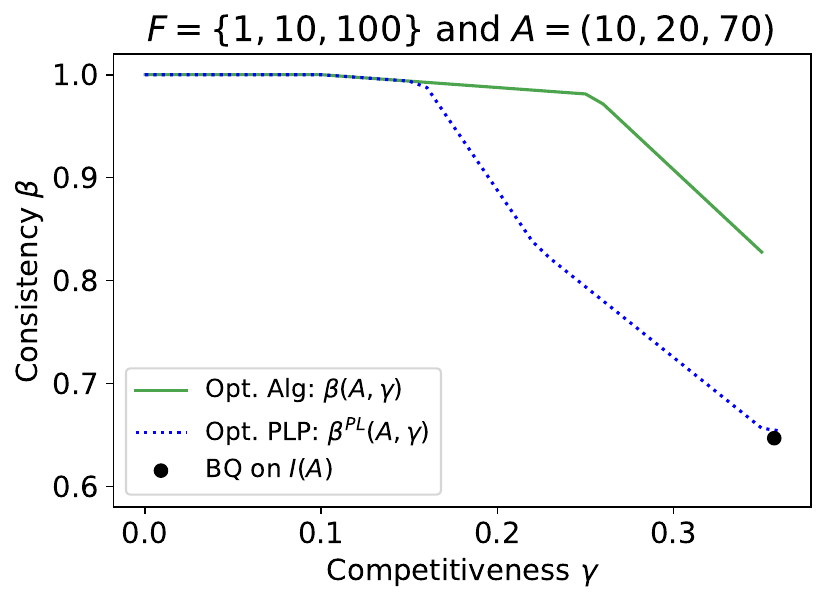}
	\end{subfigure}
	\caption{Comparing the performance frontiers of the LP-based optimal algorithm (dotted) and the optimal protection level policy (solid), with capacity $n = 100$. We use the consistency of the $c(F)$-competitive protection level policy of \citet{ball2009toward} as a benchmark.}
	\label{fig:sim-opt-vs-protec}
\end{figure}

In Figure~\ref{fig:sim-opt-vs-protec}, we plot the consistency-competitiveness performance frontiers of the LP-based optimal algorithm (Algorithm~\ref{alg:opt-alg}), represented by $\beta(A,\gamma)$, and that of the optimal protection level policy, represented by $\beta^{PL}(A, \gamma)$ and computed using Algorithm~\ref{alg:opt-protec}. We do so for two sets of fares: (i) fares that are close to each other, as represented by $F = \{1,2,4\}$ with $c(F) = 0.5$; (ii) fares that are well-separated, represented by $F = \{1,10, 100\}$ with $c(F) \approx 0.36$. Moreover, we also vary the advice by how much weight it puts on the different fare classes. These plots indicate that protection level policies are only sub-optimal when the advice is similar to the bad example discussed in Subsection~\ref{subsec:protec-level-suboptimal}. To explore this trend further, we define the notion of relative sub-optimality of protection level policies for every advice $A$ as
\begin{align*}
	RS(A) \coloneqq \max_{\gamma \in [0,c(F)]}\frac{\beta(A, \gamma) - \beta^{PL}(A, \gamma)}{\beta(A,\gamma)}\,.
\end{align*}

\begin{figure}
	\begin{subfigure}{0.45\textwidth}
		\includegraphics[width = \linewidth]{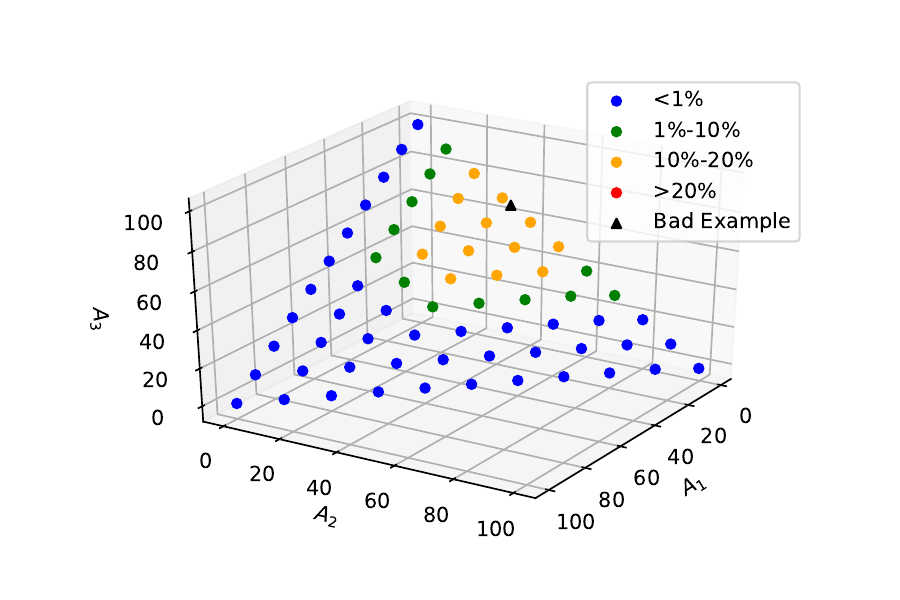}
		\caption{Close fares: $F = \{1,2,4\}$}
	\end{subfigure}%
	\begin{subfigure}{0.45\textwidth}
		\includegraphics[width = \linewidth]{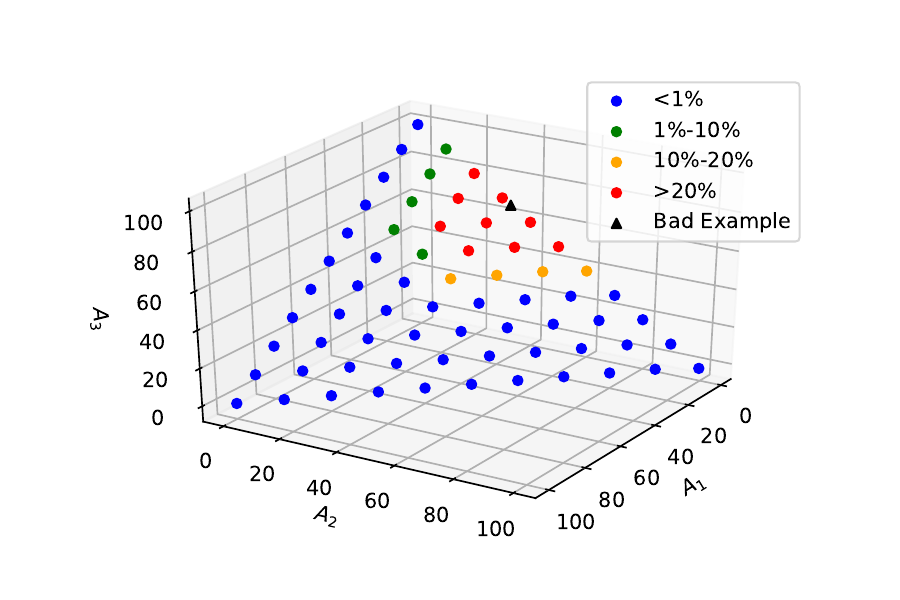}
		\caption{Separated fares: $F = \{1,10,100\}$}
	\end{subfigure}%
	\caption{Relative sub-optimality $RS(A)$ of protection level policies (Capacity $n=100$). The black triangular point denotes the advice $A = (1,33,99)$, which is an instantiation of Example~\ref{example:bad}.}
	\label{fig:grid-suboptimal}
\end{figure}

We evaluated $RS(A)$ for the grid of advice composed of all advice $A$ that satisfy $A_i\ (mod)\ 10 = 0$\footnote{With the following modification: whenever $A_1 = 0$, we  we set $A_1 = 1$ and reduce $A_3$ by 1. {\newsanti This change is introduced to make the advice} more similar to the bad example of Subsection~\ref{subsec:protec-level-suboptimal}}. We classified this grid of advice based on relative sub-optimality, represented by the different colors in Figure~\ref{fig:grid-suboptimal}. We found that the relative sub-optimality is less than $0.01$ for a large majority of the advice. Moreover, we also found that the relative sub-optimality is high for advice that makes the optimal protection level policy overcommit (see Subsection~\ref{subsec:protec-level-suboptimal}). To demonstrate this, we also plot the bad example (Example~\ref{example:bad}) of Subsection~\ref{subsec:protec-level-suboptimal} which is designed to make the optimal protection level policy overcommit. We found that the relative sub-optimality increases with the proximity to this example. Moreover, we performed this evaluation of the relative sub-optimality for two sets of fares: (i) Fares that are close to each other, as represented by $F = \{1,2,4\}$; (ii) Fares that are well-separated, represented by $F = \{1,10, 100\}$ with $c(F) = 0.357$. We found that the sub-optimality increases with the separation of the fares, but not too greatly. Finally, we also found that, across all the values of advice and fares we considered, the maximum value of the relative sub-optimality was always less than $1/3$. Our experiments suggest that the optimal protection level policy performs remarkably well on most advice, only exhibiting sub-optimality when the advice compels it to overcommit. In the next section, we continue with the comparison between the two policies with the added dimension of robustness.

\subsection{Robustness}

\begin{figure}[t!]
	\begin{subfigure}{0.33\textwidth}
		\includegraphics[width = \linewidth]{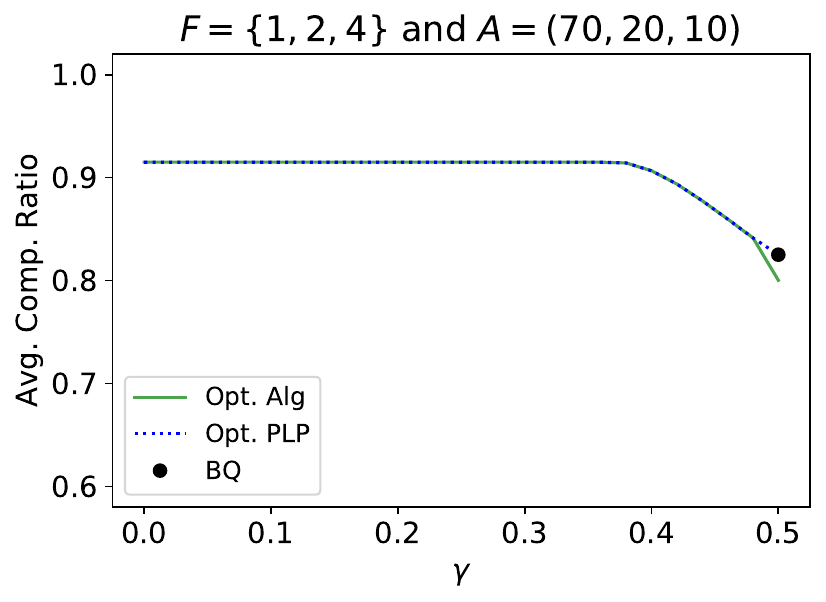}
	\end{subfigure}%
	\begin{subfigure}{0.33\textwidth}
		\includegraphics[width = \linewidth]{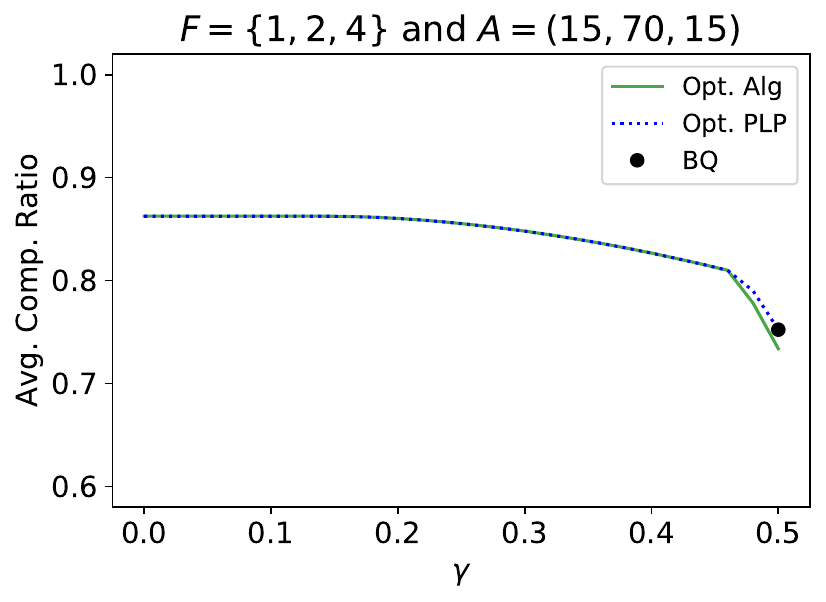}
	\end{subfigure}%
	\begin{subfigure}{0.33\textwidth}
		\includegraphics[width = \linewidth]{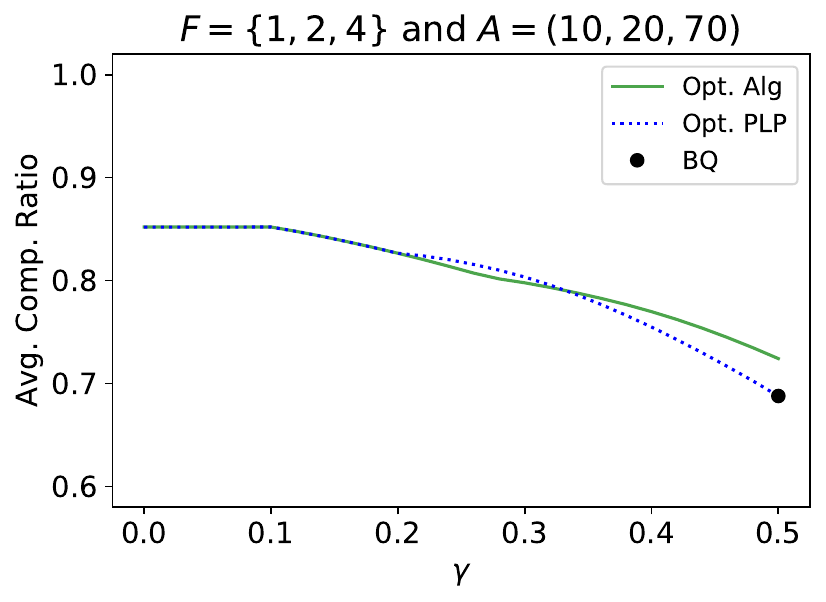}
	\end{subfigure}
	\caption{Average competitive ratio when the instances are drawn from a distribution with $v = 50\%$ relative standard deviation centered on the advice.}
	\label{fig:rob-gamma}
\end{figure}

\begin{figure}[t!]
	\begin{subfigure}{0.33\textwidth}
		\includegraphics[width = \linewidth]{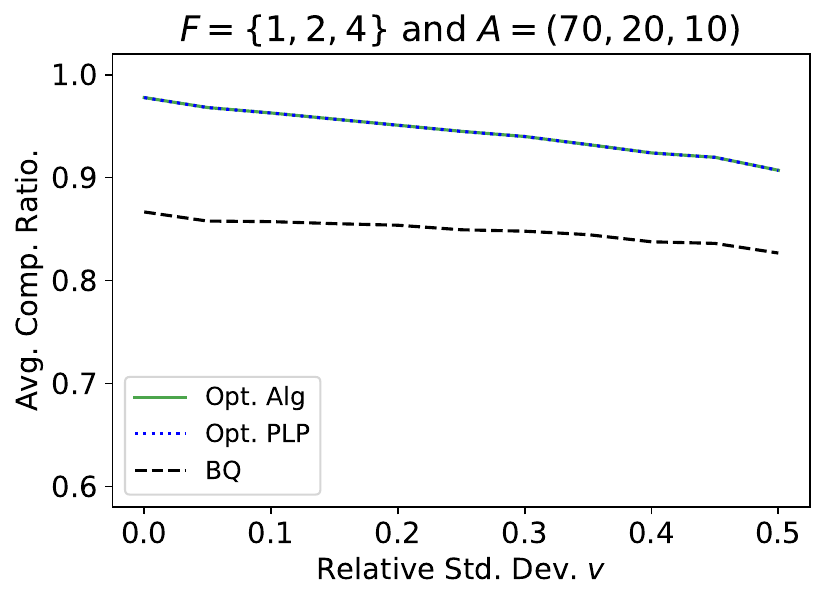}
	\end{subfigure}%
	\begin{subfigure}{0.33\textwidth}
		\includegraphics[width = \linewidth]{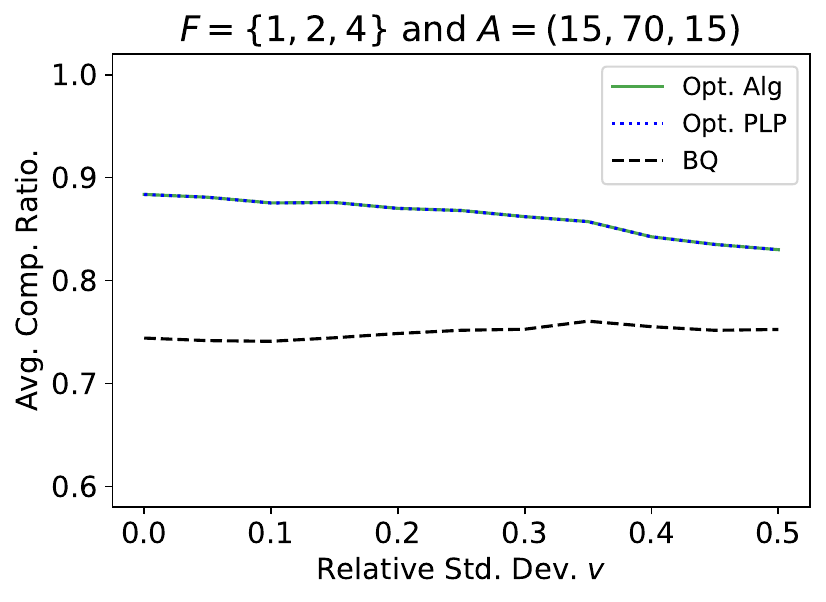}
	\end{subfigure}%
	\begin{subfigure}{0.33\textwidth}
		\includegraphics[width = \linewidth]{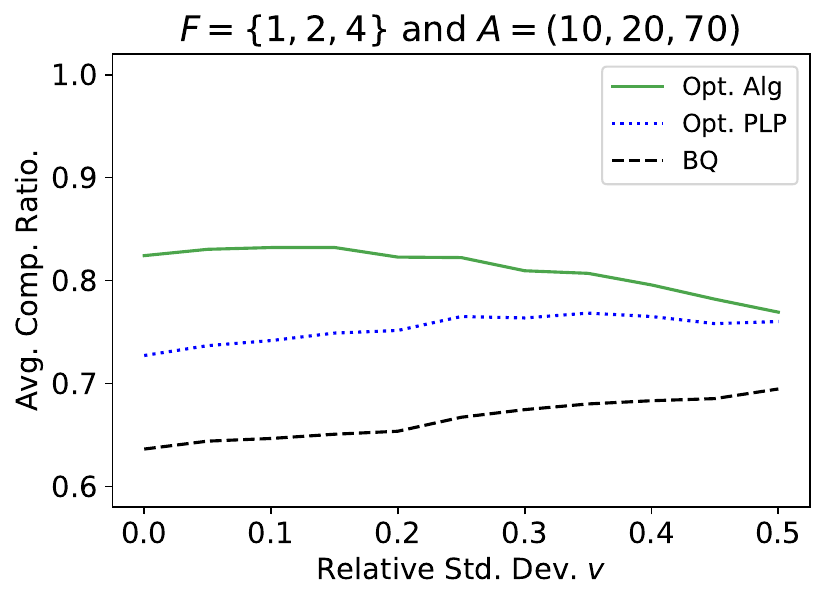}
	\end{subfigure}
	\caption{Average competitive ratio as a function of the noise in the instance-generating process as captured by the relative standard deviation $v$. Here $\gamma = 0.4$ in all plots.}
	\label{fig:v-gamma}
\end{figure}

In this subsection, we evaluate the robustness of the LP-based optimal algorithm (Algorithm~\ref{alg:opt-alg}) and the optimal protection level policy (Algorithm~\ref{alg:opt-protec}) in a variety of settings. To do so, given advice $A$, we generated random instances centered on the advice $A$ as follows: (i) For each fare class $f_i$ with $i \geq 2$, we draw $x_i \sim N(\mu_i, \sigma_i)$, where $\mu_i = A_i$ and $\sigma_i = v \cdot A_i$ for some coefficient of variation $v \in [0,1)$ (also called relative standard deviation); (ii) To get a non-negative integer value, we set $y_i = \max\{\text{floor}(x_i), 0\}$ for $i \geq 2$, and set $y_1 = 100$ to make sure that the capacity can be filled; (iv) We string together $y_i$ fares of type $f_i$ in increasing order to construct an instance $I$. We fix the fares to be $F = \{1,2,4\}$ (which implies $c(F) = 0.5$) and the capacity to be $n = 100$. Moreover, we consider three different advice given by $\{(70, 20,10), (15,70,15), (10,20,70)\}$. In our first experiment on robustness, we set the coefficient of variation $v = 0.5$ and drew 1000 instances for each advice. We computed the average competitive ratio of the algorithms on these instances as a function of $\gamma$ (Algorithm~\ref{alg:opt-alg} and Algorithm~\ref{alg:opt-protec} were given the advice as the input). The results are shown in Figure~\ref{fig:rob-gamma}. As the figure shows, even when the instances have $50\%$ relative noise compared to the advice, our algorithms continue to perform better than the policy of \citet{ball2009toward}. Moreover, we find the difference in performance between the LP-based optimal policy and the optimal protection level policy reduces in the presence of noise even for advice that makes the optimal protection level policy overcommit in the absence of noise. We note that, for advice $A = (70,20,10)$ and $A = (15,70,15)$, the performance of the LP-based optimal algorithm is identical to that of the optimal protection level policy because the optimal protection level policy is optimal in the absence of noise (see Figure~\ref{fig:sim-opt-vs-protec}), which leads \ref{LP} to return a solution that mimics the optimal protection levels. In our second experiment on robustness, we fixed $\gamma = 0.4$ and changed the noise in the instances by altering the coefficient of variation. In Figure~\ref{fig:v-gamma}, we plot the average competitive ratio (averaged over 1000 instances sampled using the aforementioned procedure) as a function of the coefficient of variation. The results indicate that the performance of our algorithms degrades gracefully with the noise in the instance generating process. It is important to note that we did not use the robust version of our LP-based optimal algorithm described in Section~\ref{sec:robustness}, which boasts stronger theoretical robustness guarantees.

\section{Conclusion and Future Work}

In this work, we look at the single-leg revenue management problem through the lens of Algorithms with Advice. We describe an optimal procedure for incorporating advice using an LP-based algorithm (Algorithm~\ref{alg:opt-alg}). Moreover, we also provide an efficient method for optimally incorporating advice into protection level policies, which is the most widely-used class of algorithms in practice. Our algorithms are simple, efficient and easy to implement. We believe that our optimal protection level policy should be relatively straightforward to implement on top of the existing infrastructure of airlines and hotels, which heavily rely on protection level policies. We hope that the algorithmic insights developed in this work can make these systems more robust to prediction errors.

We leave open the question of incorporating advice that dynamically changes with time. In practice, our algorithms can be repeatedly used with the updated advice on the remaining capacity, but our results do not capture their performance under arbitrary dynamic advice. Like many of the previous works on algorithms with advice, we also consider deterministic advice. This is motivated by the fact that most machine-learning models make point predictions. Our robustness guarantees allow us to bound the performance for distributions centered on the advice, but it would be interesting to analyze distributional advice, especially one which is updated in a Bayesian manner across time. Furthermore, we leave extensions of our results to network revenue management problems with multiple resources for future work.

\singlespace

\bibliographystyle{plainnat}
\bibliography{refs}

\appendix

\revision{
\section{Application to the Online Knapsack Problem}\label{appendix:knapsack}

Consider an online knapsack problem where the possible weights are $1 = w_1 < w_2 < \dots < w_a$ with $w_j \in \mathbb N$ and the possible item rewards are $r_1 < \dots < r_b$.  Define the corresponding set of possible fares $F = \{r_i/w_j \mid i \in [b];\ j \in [a]\}$. Extending our definition of advice for the single-leg revenue management problem, we will assume that the advice $A$ for the online knapsack problem specifies the total number of the different types of items that are predicted to form the optimal solution. Construct the corresponding advice $A'$ for single-leg revenue management problem as follows: If $k$ units of item $(r_i, w_j)$ are predicted to be a part of the optimal solution in $A$, then $k \cdot w_j$ customers with fare $r_i/w_j$ are predicted to be in $A'$. Moreover, given an instance $I$ for the online knapsack problem, construct an instance for the single-leg revenue management problem as follows: Iterate over $I$ and for each item $(r_i, w_j)$ in the instance $I$, add $w_j$ customers with fare $r_i/w_j$ to $I'$. These constructions ensure that $Opt(I) =Opt(I')$ and $Opt(A) = Opt(A')$. Therefore, it is easy to see that when we run our algorithm on $I'$ with advice $A'$ and translate the decisions to online knapsack by accepting the corresponding fraction of each item, we achieve the same revenue bounds for the online knapsack problem.

\section{Rounding Error}\label{appendix:rounding-error}

We start by observing that whenever a protection level policy based on protection levels $Q$ fractionally accepts a customer $s_t = f_p$, there is a protection level $Q_i$, for some $i \geq p$, which becomes tight, i.e., the protection level policy no longer accepts a non-zero fraction of customers with fare $f_i$ or less. Since there are $m$ protection levels, this implies that each protection level policy fractionally accepts at most $m$ customers. Consequently, Algorithm~\ref{alg:opt-protec} fractionally accepts at most $m$ customers. Moreover, as Algorithm~\ref{alg:opt-alg} switches between at most 2 protection levels, it fractionally accepts at most $2m$ customers.

Our algorithms can be simply modified to make binary accept reject decisions at the cost of only a minor degradation in performance. We show this for Algorithm~\ref{alg:opt-alg}, a similar argument holds for Algorithm~\ref{alg:opt-protec}. Fix capacity $n$, advice $A$, required level of competitiveness $\gamma \in [0, c(F)]$ and instance $I$. For modified capacity $n' \leq n$,
\begin{itemize}
	\item let $Alg(I, n')$ denote the performance of Algorithm~\ref{alg:opt-alg} on instance $I$ when the input capacity is $n'$, required level of competitiveness is $(n'/n) \cdot \gamma$, and we leave the advice $A$ unchanged (In particular, $n$ changes to $n'$ and $\gamma$ changes to $(n'/n) \cdot \gamma$ in \ref{LP}, but $A_i$ remain the same in the algorithm),
	\item let $Opt(I, n')$ denote the optimal total reward on instance $I$ when the capacity is $n'$.
\end{itemize}

Now, consider the following modification of Algorithm~\ref{alg:opt-alg} that makes binary decisions: Run Algorithm~\ref{alg:opt-alg} with a modified input capacity $n'=n - 2m$, and whenever it accepts a non-zero fraction of a customer, completely accept that customer. Given the modified capacity of $n- 2m$ and our earlier observation that Algorithm~\ref{alg:opt-alg} fractionally accepts at most $2m$ customers, we get that this modified algorithm accepts at most $n - 2m + 2m = n$ customers.
Finally, observe that the average of the top $n - 2m$ fares is larger than the average of the top $n$ highest fares, and consequently, we get $Opt(I, n-2m)/(n-2m) \geq Opt(I,n)/n$. Letting $Alg'(I)$ denote the total reward of the modified algorithm on instance $I$, we can bound the competitive ratio of the modified algorithm from below as follows:
\begin{align*}
	\frac{Alg'(I)}{Opt(I,n)} &\geq \frac{Alg(I, n-2m)}{Opt(I,n)} \geq \frac{n - 2m}{n}\cdot\frac{Alg(I, n-2m)}{Opt(I,n-2m)} \geq \frac{Alg(I, n-2m)}{Opt(I,n-2m)} - \frac{2m}{n}\,.
\end{align*}

It is straightforward to see that any solution $(\beta, x, y(k))$ to \ref{LP} with capacity $n$ can be rescaled to obtain a solution $\{(n-2m)/n\} \cdot (\beta, x,y(k))$ for \ref{LP} with capacity $n-2m$ and required level of competitiveness $\{(n-2m)/n\} \cdot \gamma$. Therefore, Theorem~\ref{thm:lower-bound} implies
\begin{align*}
	 \frac{Alg(I, n-2m)}{Opt(I,n-2m)} \geq \frac{n-2m}{n} \cdot \frac{Alg(I,n)}{Opt(I,n)} \geq \frac{Alg(I,n)}{Opt(I,n)} - \frac{2m}{n} \quad \forall\ I \in S(A)\,.
\end{align*}
Moreover, for all other instances $I$, we have 
\begin{align*}
	\frac{Alg(I, n-2m)}{Opt(I,n-2m)} \geq \frac{n-2m} n \cdot \gamma \geq \gamma - \frac {2m} n\,.
\end{align*}
Hence, the consistency and competitiveness of the modified algorithm are at most $4m/n$ less than the corresponding quantities for Algorithm~\ref{alg:opt-alg}.
}
\section{Missing Proofs of Section~\ref{sec:pareto}}

We begin by proving that every protection level policy earns the least amount of revenue when the customers arrive in increasing order of fares. This fact guides our intuition and we use it repeatedly throughout the paper.

\begin{lemma}\label{lemma:protec-inc-order}
	Let $I$ be an instance and let $\tilde{I}$ be the instance obtained by the reordering the customers of $I$ to arrive in increasing order. Then, for any protection level policy, the revenue earned from $I$ is greater than or equal to the revenue earned from $\tilde{I}$.
\end{lemma}

\begin{proof}
	Consider a protection level policy with protection levels $Q$. If $I$ has customers arriving in increasing order, then $I = \tilde{I}$ and the lemma holds trivially. Suppose $I = \{s_t\}_{t=1}^T$ and there exists $a < T$ and $k,\ell \in [m]$ such that $s_a = f_\ell > f_k = s_{a+1}$. Define another instance $I' = \{s'_t\}_{t=1}^T$ in which $s_a$ and $s_{a+1}$ are flipped, i.e., $s'_a = s_{a+1}$, $s'_{a+1} = s_a$ and $s'_t = s_t$ for all $t \notin \{a, a+1\}$. To prove the lemma, it suffices to show that the protection level policy $Q$ obtains (weakly) more revenue from $I$ than $I'$. Note that the instances are identical till $t = a-1$.  For all $t \leq a-1$, let $x_t$ denote the fraction of customer $t$ accepted by the protection level policy $Q$ on instance $I$ and $I'$. Moreover, define
	\begin{align*}
		r_j = \min_{k \geq j} \left\{ Q_k - \sum_{i=1}^{a-1} x_i\mathds{1}(s_i \leq f_k) \right\}
	\end{align*}
	to be the amount of a customer with fare $f_j$ that can be accepted at time $a$ without violating any of the protection levels. It is easy to see that $r_k \leq r_\ell$ because $k < \ell$. Next, consider the following cases:
	\begin{itemize}
		\item $r_\ell < 1$ : Then, $Q$ accepts $r_\ell$ fraction of $s_a$ and does not accept any fraction of $s_{a+1}$ in instance $I$. Moreover, it accepts at most $r_\ell$ in total from both $s'_a$ and $s'_{a+1}$ in instance $I'$.
		\item $r_\ell \geq 2$: Then, $Q$ accepts all of $s_a$ in $I$ and all of $s'_{a+1}$ in $I'$. Moreover, it accepts $\min\{1, r_k\}$ fraction of $s_{a+1}$ in $I$ and $s'_a$ in $I'$.
		\item $1 \leq r_\ell <2$ and $r_k \leq r_\ell - 1$: Then, $Q$ accepts all of $s_a$ and $r_k$ fraction of $s_{a+1}$ in instance $I$. Moreover, it accepts at most $r_k$ fraction of customer $s'_a$ and all of customer $s'_{a+1}$ in instance $I'$.
		\item $1 \leq r_\ell <2$ and $r_k > r_\ell - 1$: Then, $Q$ accepts all of $s_a$ and $r_\ell - 1$ fraction of $s_{a+1}$ in instance $I$. Moreover, it accepts at most $\min\{1, r_k\}$ fraction of customer $s'_a$ and at most $r_\ell - \min\{1, r_k\}$ fraction of customer $s'_{a+1}$ in instance $I'$.
	\end{itemize}
	Since $f_k < f_\ell$, it is easy to see that $Q$ earns more revenue from $I$ than $I'$ in all these cases.
\end{proof}

\revision{
\subsection{Proof of Theorem~\ref{thm:lower-bound}}\label{appendix:opt-alg}

In this section, we analyze Algorithm~\ref{alg:opt-alg}, and show that it is $\beta^*$-consistent on advice $A$ and $\gamma$-competitive. Let $\tau$ denote the first time step $t$ in which $\sum_{i=j}^h r_i > \sum_{i=j}^h \{A_i - x_i\}$ for some $j \geq \ell + 1$ in Step 2, i.e., the `If' condition in Step 2 is satisfied and the trigger changes to $\delta =1$. Set $\tau = \infty$ if the condition is never satisfied. Let $q_j^*$, and $r_j^*$ denote the value $q_j$ and $r_j$ respectively at the end of time $T$. Moreover, let $q(\tau)_i$ and $r(\tau)_i$ denote the value of $q_i$ and $r_i$ respectively at the beginning of time $\tau$. We will use $a_j^*$ and $a(\tau)_j$ to denote the total number of customers with fare $f_j$ that arrive by the end of time $T$ and the start of time $\tau$ respectively. The following lemma notes some important properties of the protection levels $R(k)$ and establishes that the `While' loop in step 2 of Algorithm~\ref{alg:opt-alg} always terminates with an index $k \geq s$.

\begin{lemma}
	Algorithm~\ref{alg:opt-alg} satisfies the following properties:
	\begin{enumerate}
		\item $Q_j' \leq R(k)_j$ for all $j \leq k$.
		\item The `If' condition in step 2 is only satisfied at time $\tau$.
		\item The value of $k$ increases by at least 1 in every iteration of the `While' loop in step 2.
		\item If $\tau \leq T$, then $q(\tau)_j \leq R(k^*)_j$ for all $j \in [m]$.
	\end{enumerate}
\end{lemma}

\begin{proof}
	Property (1) follows directly from the definition of protection level $R(k)$.
	Property (2) follows from the fact that once the trigger $\delta$ is set to be 1, then the `If' condition in step 2 is never satisfied.
	Property (3) follows from the fact that Property (1), in conjunction with $q(\tau)_j \leq Q_j'$ for all $j \in [m]$, implies that the smallest index $i$ that satisfies $q(\tau)_i > R(k)_i$ is strictly greater than $k$ for all $k \in [m]$. Hence, the value of $k$ increases by at least 1 in every iteration of the loop in step~2.
	Property (4) follows because the termination condition of the loop in step~2 implies that $q(\tau)_j \leq R(k^*)_j$ for all $j \in [m]$.
\end{proof}

Before moving onto the consistency and competitiveness properties of Algorithm~\ref{alg:opt-alg}, we note a useful algebraic manipulation that will find repeated use in our proof.

\begin{lemma}\label{lemma:rev-rewrite}
	For any collection of $m$ real numbers $\{\lambda_1, \lambda_2, \dots, \lambda_m\}$, the following statement holds for all $1 \leq k \leq i \leq m$:
	\begin{align*}
		\sum_{j=k+1}^i (\lambda_j - \lambda_{j-1}) \cdot (f_j - f_k) = \sum_{j=k+1}^i (\lambda_i - \lambda_{j-1}) \cdot (f_j - f_{j-1})\,.
	\end{align*}
\end{lemma}
\begin{proof}
	The lemma follows from changing the order of the sums:
	\begin{align*}
		\sum_{j=k+1}^i (\lambda_j - \lambda_{j-1}) \cdot (f_j - f_k) &= \sum_{j=k+1}^i \sum_{p=k+1}^j (\lambda_j - \lambda_{j-1}) \cdot (f_p - f_{p-1})\\
		&= \sum_{p=k+1}^i \sum_{j=p}^i (\lambda_j - \lambda_{j-1}) \cdot (f_p - f_{p-1})\\
		&= \sum_{p=k+1}^i (\lambda_i - \lambda_{p-1}) \cdot (f_p - f_{p-1})\,,
	\end{align*}
	as required.
\end{proof}

The next lemma shows that the trigger is never invoked as long as the advice can potentially be realized.

\begin{lemma}\label{lemma:better-trigger}
	On any instance $I = \{s_t\}_{t=1}^T$ such that $a^*_j = \sum_{t=1}^T \mathds{1}(s_t = f_j) \leq A_j$ for all $j \geq \ell+1$, we have $\delta =0$ throughout the run of Algorithm~\ref{alg:opt-alg}.
\end{lemma}
\begin{proof}
	For contradiction, suppose $\tau \leq T$ and $s_\tau = f_p$ for some $p \geq \ell+1$. For $w^* = \max\{w \in [0,1] \mid q(\tau)_j + w \leq Q_j'\ \forall j \geq p\}$, let $h \geq p$ be the largest index such that $q(\tau)_h + w^* = Q'_h$. Such an index always exists because $w^*< 1$ by definition of $\tau$. Moreover, the definition of $\tau$ also implies the existence of an index $g \in [\ell+1, p]$ such that $(1 - w^*) + \sum_{i=g}^h r(\tau)_i > \sum_{i=g}^h \{A_i - x_i\}$. Since $q(\tau)_{g-1} \leq Q'_{g-1}$ by definition of protection level policies, we get that the algorithm accepted at least $q(\tau)_h - q(\tau)_{g-1} \geq Q'_h - w^* - Q'_{g-1}$ customers with fares $\{f_g ,\dots, f_h\}$ before time $\tau$. Since the $\tau$-th customer has fare $f_p$ such that $g \leq p \leq h$, we get that the total number of customers with fare $\{f_g ,\dots, f_h\}$ that arrive in $I$ is at least $1 + Q'_h - w^* - Q'_{g-1} + \sum_{i=g}^h r(\tau)_i$. Moreover, using $Q'_h - Q'_{g=1} = \sum_{i=g}^h x_i$, we can bound this from below:
	\begin{align*}
		Q'_h - Q'_{g-1} + (1 - w^*) + \sum_{i=g}^h r(\tau)_i > \sum_{i=g}^h x_i + \sum_{i=g}^h \{A_i - x_i\} = \sum_{i=g}^h A_i \,,
	\end{align*}
	which contradicts $a^*_j \leq A_j$ for all $j \in [g,h]$. Therefore, we have $\tau = \infty$ and $\delta = 0$ throughout the run of Algorithm~\ref{alg:opt-alg} on $I$.
\end{proof}

Next, we show that, when the instance conforms to the advice, Algorithm~\ref{alg:opt-alg} achieves at least a $\beta^*$ fraction of $Opt(A)$ in terms of revenue.

\begin{lemma}\label{lemma:alg-consistent}
	On any instance $I \in S(A)$, Algorithm~\ref{alg:opt-alg} achieves a revenue that is greater than or equal to $\beta^* \cdot Opt(A)$.
\end{lemma}

\begin{proof}
	Consider an instance $I \in S(A)$. Then, $a_j \leq A_j$ for all $j \geq \ell+1$ throughout the run of Algorithm~\ref{alg:opt-alg}. Therefore, Lemma~\ref{lemma:better-trigger} implies that Algorithm~\ref{alg:opt-alg} always accepts customers according to protection levels $Q^* =Q'$, i.e. $q_i \leq Q_i'$ for all $i \in [m]$ is satisfied at all times. Since $I \in S(A)$, we have $\sum_{t=1}^T \mathds{1}(s_t = f_i) = A_i\ \forall i \geq \ell+1$ and $\sum_{t=1}^T \mathds{1}(s_t = f_\ell) \geq A_\ell$. Combining this with the fifth constraint of the \ref{LP} ($x_i \leq N_i$ for all $i \in [m]$) implies that, for all $j \in [m]$, Algorithm~\ref{alg:opt-alg} accepts $q^*_m - q^*_{j-1} \geq Q_m' - Q_{j-1}'$ customers whose fare type lies in the set $\{f_{j}, \dots, f_m\}$. Observe that Algorithm~\ref{alg:opt-alg} accepts $q_j^* - q_{j-1}^*$ customers with fare $f_j$ for all $j \in [m]$. This allows us to establish the following lower bound on its revenue:
	\begin{align*}
		\sum_{j=1}^m f_j (q_j^* -q_{j-1}^*) &= \sum_{p=1}^m (q^*_m - q^*_{p-1})(f_p - f_{p-1}) \tag{Lemma~\ref{lemma:rev-rewrite}}\\
		&\geq \sum_{p=1}^m (Q'_m - Q'_{p-1})(f_p - f_{p-1})\\
		&= \sum_{j=1}^m (Q'_j - Q'_{j-1}) \cdot f_j  \tag{Lemma~\ref{lemma:rev-rewrite}}\\
		&= \sum_{j=1}^m x_j \cdot f_j\\
		&\geq \beta^* \cdot Opt(A) \tag{\ref{LP}}
	\end{align*}
	Thus, we have shown that, on any instance $I \in S(A)$, Algorithm~\ref{alg:opt-alg} achieves a revenue that is greater than or equal to $\beta^* \cdot Opt(A)$.
\end{proof}

Next, we establish $\gamma$-competitiveness, starting with instances which do not trigger the change in protection levels.

\begin{lemma}\label{lemma:alg-prefix}
	On any instance $I$ such that $\tau = \infty$, Algorithm~\ref{alg:opt-alg} achieves a revenue that is greater than or equal to $\gamma \cdot Opt(I)$.
\end{lemma}

\begin{proof}
	As $\delta = 0$ throughout the run of Algorithm~\ref{alg:opt-alg} on instance $I$, Algorithm~\ref{alg:opt-alg} always accepts customers according to protection levels $Q^* = Q'$. Let $i \in [m]$ be the largest index such that $q_i^* = Q'_i$ (If no such index exists, then Algorithm~\ref{alg:opt-alg} has accepted every customer and achieved the optimal total reward of $Opt(I)$). Then, Algorithm~\ref{alg:opt-alg} accepts every customer in $I$ with fare strictly greater than $f_i$, i.e., $r_j^* = 0$ for $j > i$ and receives a revenue of $\sum_{j=i+1}^m f_j a^*_j$ from customers lying in the set $\{f_{i+1}, \dots f_m\}$. For all $j \leq i$, it accepts $q^*_j - q^*_{j-1}$ customers from fare class $f_j$. Therefore, the total reward of the algorithm is given by
	\begin{align}\label{eqn:prefix-inter-1}
		Alg(I) \coloneqq \left\{\sum_{j=1}^i (q^*_j - q^*_{j-1}) \cdot f_j \right\} + \sum_{j=i+1}^m a_j^* \cdot f_j\,.
	\end{align}

	When $i > \ell$, $Opt(I)$ is composed of at most $a_j^*$ customers with fares $f_j$ for all $j \geq \ell+1$ and the remainder with fares less than or equal to $f_\ell$. Moreover, $a^*_j = q^*_j - q^*_{j-1} + r^*_j$ because, for each fare class, the total amount of customers accepted by the algorithm and the total amount of customers rejected by the algorithm together equal the total number of customers that arrived. Consequently, we get
	\begin{align*}
		Opt(I) &\leq n \cdot f_\ell + \sum_{j = \ell+1}^i  a_j^* \cdot (f_j - f_\ell) + \sum_{j=i+1}^m a^*_j \cdot (f_j - f_\ell)\\
		&\leq  n \cdot f_\ell + \sum_{j = \ell+1}^i  \left( (q^*_j - q^*_{j-1}) + r_j^* \right) \cdot (f_j - f_\ell) + \sum_{j=i+1}^m a^*_j \cdot (f_j - f_\ell)  \,.
	\end{align*}

	Since $\delta = 0$ throughout the run of the algorithm, the definition of $i$ implies that $\sum_{g=j}^i r_g^* \leq \sum_{g=j}^i \{A_g - x_g\}$ for all $j \in [\ell+1, i]$. This allows us to write
	\begin{align}\label{eqn:prefix-inter-2}
		Opt(I) &\leq  n \cdot f_\ell + \sum_{j = \ell+1}^i r_j^* \cdot (f_j - f_\ell) + \sum_{j = \ell+1}^i  (q^*_j - q^*_{j-1}) \cdot (f_j - f_\ell) + \sum_{j=i+1}^m a^*_j \cdot (f_j - f_\ell) \nonumber \\
		&\leq n \cdot f_\ell + \sum_{j = \ell+1}^i (A_j - x_j) \cdot (f_j - f_\ell) + \sum_{j = \ell+1}^i  (q^*_j - q^*_{j-1}) \cdot (f_j - f_\ell) + \sum_{j=i+1}^m a^*_j \cdot (f_j - f_\ell) \nonumber \\
		&\leq n \cdot f_\ell + \sum_{j = \ell+1}^i A_j \cdot (f_j - f_\ell) + \sum_{j = \ell+1}^i  (q^*_j - q^*_{j-1} - x_j) \cdot (f_j - f_\ell)  + \sum_{j=i+1}^m a^*_j \cdot f_j \nonumber \\
		&= \left(n - \sum_{j = \ell+1}^i A_j \right) \cdot f_\ell + \sum_{j = \ell+1}^i A_j \cdot f_j + \sum_{j = \ell+1}^i  (q^*_j - q^*_{j-1} - x_j) \cdot (f_j - f_\ell)  + \sum_{j=i+1}^m a^*_j \cdot f_j \nonumber \\
		&= Opt(I(A,i)) + \sum_{j = \ell+1}^i  (q^*_j - q^*_{j-1} - x_j) \cdot (f_j - f_\ell)  + \sum_{j=i+1}^m a^*_j \cdot f_j
	\end{align}
	
	When $i \leq \ell$, we also have $Opt(I) \leq n \cdot f_i + \sum_{j=i+1}^m a_j^* \cdot f_j = Opt(I(A,i)) + \sum_{j=i+1}^m a_j^* \cdot f_j$ from the definition of $I(A,i)$. Therefore, combining \eqref{eqn:prefix-inter-1} and \eqref{eqn:prefix-inter-2}, we get
	\begin{align*}
		Alg(I) - \gamma \cdot Opt(I) &\geq \left\{\sum_{j=1}^i (q^*_j - q^*_{j-1}) \cdot f_j \right\} - \gamma \cdot Opt(I(A,i)) - \gamma \cdot \sum_{j = \ell+1}^i  (q^*_j - q^*_{j-1} - x_j) \cdot (f_j - f_\ell)\\
		&\geq \left\{\sum_{j=1}^i (q^*_j - q^*_{j-1}) \cdot f_j \right\} - \sum_{j=1}^i x_j \cdot f_j - \gamma \cdot \sum_{j = \ell+1}^i  (q^*_j - q^*_{j-1} - x_j) \cdot (f_j - f_\ell)\\
		&\geq \left\{\sum_{j=1}^i (q^*_j - q^*_{j-1} - x_j) \cdot f_j \right\} - \sum_{j = \ell+1}^i  (q^*_j - q^*_{j-1} - x_j) \cdot (f_j - f_\ell)\\
		&= \left\{\sum_{j=1}^i \left(q^*_i - q^*_{j-1} - \sum_{k=j}^i x_k \right) \cdot (f_j - f_{j-1}) \right\} - \sum_{j = \ell+1}^i  \left(q^*_i - q^*_{j-1} - \sum_{k=j}^i x_k \right) \cdot (f_j - f_{j-1})\\
		&= \sum_{j=1}^{\min\{i,\ell\}} \left(q^*_i - q^*_{j-1} - \sum_{k=j}^i x_k \right) \cdot (f_j - f_{j-1})\\
		&= \sum_{j=1}^{\min\{i,\ell\}} \left(q^*_i - q^*_{j-1} - (Q_i' - Q_{j-1}') \right) \cdot (f_j - f_{j-1})\,,
	\end{align*}
	where the first equality holds due to Lemma~\ref{lemma:rev-rewrite}. Since $q^*_i = Q_i'$ and $q_j^* \leq Q_j'$ for all $j \in [m]$, we get $Alg(I) - \gamma \cdot Opt(I) \geq 0$ as required.
\end{proof}

The following lemma covers the remaining case needed to establish $\gamma$-competitiveness of Algorithm~\ref{alg:opt-alg}.

\begin{lemma}\label{lemma:alg-competitive}
	On any instance $I$ with $\tau \leq T$, Algorithm~\ref{alg:opt-alg} achieves a revenue greater than or equal to $\gamma \cdot Opt(I)$.
\end{lemma}

\begin{proof}
	Let $h(\tau) \coloneqq \max\{i \in [m] \mid q(\tau)_i = Q'_i\}$ (here $h(\tau) =0$ if no such index exists), i.e., $h(\tau)$ is the largest non-zero value of $h$ attained before the start of iteration $\tau$. In this proof, we will repeatedly use the following observation which follows from the definition of $h(\tau)$: Algorithm~\ref{alg:opt-alg} accepts every customer with fare strictly greater than $h(\tau)$ that arrives before time $\tau$ (not including $\tau$). Let $i$ be the largest index for which $q^*_i = R(k^*)_i$ (assume $i= 0$ if no such index exists). We will break up the proof into three cases predicated on the relationship between $k^*$, $h(\tau)$ and $i$ for instance $I$.
	
	\textbf{Case I:} $k^* \geq h(\tau)$ and $i > k^*$. First, observe that the definition of $i$ and $i \geq h(\tau)$ imply that Algorithm~\ref{alg:opt-alg} accepts every customer with fare strictly greater than $f_i$. Furthermore, for all $j \leq i$, Algorithm~\ref{alg:opt-alg} accepts $q^*_i - q^*_{j-1} \geq R(k^*)_i - R(k^*)_{j-1}$ customers with fare lying in the set $\{f_j, \dots, f_i\}$. Therefore, the revenue that Algorithm~\ref{alg:opt-alg} obtains from customers with fare $f_i$ or lower is given by
		\begin{align*}
			\sum_{j=1}^i f_j(q^*_j - q^*_{j-1}) &= 	\sum_{p=1}^i (q^*_i - q^*_{p-1})(f_p - f_{p-1}) \tag{Lemma~\ref{lemma:rev-rewrite}}\\
			&\geq \sum_{p=1}^i (R(k^*)_i - R(k^*)_{p-1})(f_p - f_{p-1})\\
			&= \sum_{j=1}^i (R(k^*)_j - R(k^*)_{j-1}) \cdot f_j \tag{Lemma~\ref{lemma:rev-rewrite}}\\
			&= \sum_{j=1}^{k^*} x_j\cdot f_j + \sum_{j=1}^i y(k^*)_j \cdot f_j \\
			&\geq \gamma \cdot Opt(I(A,k^*) \oplus I(F,i)) \tag{\ref{LP}}\\
			&\geq \gamma \cdot n \cdot f_i.
		\end{align*}
	Thus Algorithm~\ref{alg:opt-alg} achieves revenue at least  $\gamma\cdot n \cdot f_i + \sum_{j=i+1}^m a^*_j \cdot f_j$.
	Since $Opt(I)$ can be at most $n \cdot f_i + \sum_{j=i+1}^m a^*_j \cdot f_j$, we get that Algorithm~\ref{alg:opt-alg} achieves a revenue greater than or equal to $\gamma \cdot Opt(I)$ in this case.
	
	\textbf{Case II:} $k^* > h(\tau)$ and $i \leq k^*$. If $k^* = 1$ and $h(\tau) = 0$, then Algorithm~\ref{alg:opt-alg} takes the same decisions as the protection level policy based on $R(1)$, which is $\gamma$-competitive. To see this, observe that Lemma~\ref{lemma:protec-inc-order} implies that the worst-case instance for any protection policy is attained on $I(F,i)$ for some $i \in [m]$. Moreover, the protection level based on $R(1)$ attains a total reward of $x_1 \cdot f_1 + \sum_{j=1}^i y(1)_j \cdot f_j$ on $I(F,i)$. Next, observe that $Opt(I(A,1) \oplus I(F,i)) = Opt(I(F,i))$ for all $i \in [m]$. Therefore, constraint~\ref{constraint:comp} implies that Algorithm~\ref{alg:opt-alg} is $\gamma$-competitive on all instances.
	
	Assume that either $h(\tau) > 0$ or $k^* > 1$. As $k^* > h(\tau)$, both inequalities imply $k^* > 1$. By our choice of $k^*$, there exists $k \geq h(\tau)$ such that $k^* = \min\{j \in [m] \mid q(\tau)_j > R(k)_j\}$. Therefore, $q(\tau)_j \leq R(k)_j$ for all $j < k^*$. Thus, for all $j \leq k^*$, we have $q(\tau)_{k^*} - q(\tau)_{j-1} \geq R(k)_{k^*} - R(k)_{j-1}$ customers with fare lying in the set $\{f_j, \dots, f_{k^*}\}$. Therefore, the revenue that Algorithm~\ref{alg:opt-alg} obtains from customers with fare $f_{k^*}$ or lower is at least
		\begin{align*}
			\sum_{j=1}^{k^*} f_j(q(\tau)_j - q(\tau)_{j-1}) &= \sum_{j=1}^{k^*} \sum_{p=1}^j (f_p - f_{p-1})(q(\tau)_j - q(\tau)_{j-1})\\
			&= 	\sum_{p=1}^{k^*} (q(\tau)_{k^*} - q(\tau)_{p-1})(f_p - f_{p-1})\\
			&\geq \sum_{p=1}^{k^*} (R(k)_{k^*} - R(k)_{p-1})(f_p - f_{p-1})\\
			&= \sum_{j=1}^{k^*} (R(k)_j - R(k)_{j-1}) \cdot f_j \tag{Lemma~\ref{lemma:rev-rewrite}}\\
			&= \sum_{j=1}^k x_j \cdot f_j + \sum_{j=1}^{k^*} y(k)_j \cdot f_j \\
			&\geq \gamma \cdot Opt(I(A,k) \oplus I(F,k^*)) \tag{\ref{LP}}\\
			&\geq \gamma \cdot n \cdot f_{k^*}
		\end{align*}
	 Furthermore, note that Algorithm~\ref{alg:opt-alg} accepts every customer with fare strictly greater than $f_{k^*}$ in $I$, yielding a reward of $\sum_{j=k^*+1}^m a^*_j \cdot f_j$. On the other hand, $Opt(I)$ can be at most $n \cdot f_{k^*}  + \sum_{j=k^*+1}^m a^*_j \cdot f_j$, thereby implying that Algorithm~\ref{alg:opt-alg} achieves a revenue greater than or equal to $\gamma \cdot Opt(I)$ in this case.
	
	\textbf{Case III:} $k^* = h(\tau)$ and $i \leq k^*$. Note that Algorithm~\ref{alg:opt-alg} accepts $(q_j^*- q_{j-1}^*)$ customers with fare $f_j$ total, and $(q(\tau)_j - q(\tau)_{j-1})$ of the customers with fare $f_j$ by the start of time step $\tau$. Hence, the revenue that Algorithm~\ref{alg:opt-alg} obtains from customers that satisfy one of the following conditions:
		\begin{itemize}
			\item Fare is less than or equal to $f_{k^*}$ and arrives before start of time $\tau$
			\item Fare is less than or equal to $f_i$	
		\end{itemize}
		is at least $\sum_{j=1}^i f_j(q^*_j - q^*_{j-1}) + \sum_{j= i+1}^{k^*} f_j (q(\tau)_j - q(\tau)_{j-1})$.
	
	
	Note that the definition of $i$ implies $q_i^* = R(k^*)_i$ and $q_j^* \leq  R(k^*)_j$ for all $j<i$, and consequently $q_i^* - q^*_{j-1} \geq R(k^*)_i - R(k^*)_{j-1}$ for all $j \leq i$. Hence,
		\begin{align*}
			&\sum_{j=1}^i f_j(q^*_j - q^*_{j-1}) + \sum_{j= i+1}^{k^*} f_j (q(\tau)_j - q(\tau)_{j-1})\\
			= &\sum_{p=1}^i (q_i^* - q_{p-1}^*)(f_p - f_{p-1}) + \sum_{p=1}^{k^*} (q(\tau)_{k^*} - q(\tau)_{p-1})(f_p - f_{p-1}) \tag{Lemma~\ref{lemma:rev-rewrite}}\\
			\geq & \sum_{p=1}^i (R(k^*)_i - R(k^*)_{p-1})(f_p - f_{p-1}) + \sum_{p=i+1}^{k^*} (Q'_{k^*} - Q'_{p-1})(f_p - f_{p-1})\\
			 &+ \sum_{p=i+1}^{k^*} (q(\tau)_{k^*} - q(\tau)_{p-1}) - (Q'_{k^*} - Q'_{p-1}))(f_p - f_{p-1}) \\
			= & \sum_{j=1}^i (R(k^*)_j - R(k^*)_{j-1})f_j + \sum_{j=i+1}^{k^*} (Q'_{j} - Q'_{j-1}) f_j\tag{Lemma~\ref{lemma:rev-rewrite}}\\
			  &+ \sum_{p=i+1}^{k^*} (q(\tau)_{k^*} - q(\tau)_{p-1}) - (Q'_{k^*} - Q'_{p-1}))(f_p - f_{p-1}) \\
			= & \sum_{j=1}^i (x_j + y(k^*)_j) \cdot f_j + \sum_{j=i+1}^{k^*} x_j \cdot f_j + \sum_{p=i+1}^{k^*} (q(\tau)_{k^*} - q(\tau)_{p-1}) - (Q'_{k^*} - Q'_{p-1}))(f_p - f_{p-1})\\
			= &\sum_{j=1}^i y(k^*)_i \cdot f_j + \sum_{j=1}^{k^*} x_j \cdot f_j + \sum_{p=i+1}^{k^*} (q(\tau)_{k^*} - q(\tau)_{p-1}) - (Q'_{k^*} - Q'_{p-1}))(f_p - f_{p-1})\\
			\geq &\ \gamma \cdot Opt(I(A,k^*) \oplus I(F,i)) + \sum_{p=i+1}^{k^*} (q(\tau)_{k^*} - q(\tau)_{p-1}) - (Q'_{k^*} - Q'_{p-1}))(f_p - f_{p-1})
		\end{align*}


	Furthermore, Algorithm~\ref{alg:opt-alg} accepts every customer with fare in the set $\{f_{i+1}, \dots, f_{s}\}$ who arrives after the start of time $\tau$ yielding a revenue of $\sum_{j=i+1}^{k^*} (a^*_j - a(\tau)_j) \cdot f_j$  and always accepts every customer with fare strictly greater than $f_{k^*}$ yielding a revenue of $\sum_{j=k^*+1}^m a^*_j \cdot f_j$. Therefore, the total reward collected by the algorithm is given by
	\begin{align}\label{eqn:comp-case-3-inter-1}
		Alg(I) \geq &\ \gamma \cdot Opt(I(A,k^*) \oplus I(F,i)) + \sum_{p=i+1}^{k^*} (q(\tau)_{k^*} - q(\tau)_{p-1} - (Q'_{k^*} - Q'_{p-1})) (f_p - f_{p-1}) \nonumber\\
		&+ \sum_{j=i+1}^{k^*} (a^*_j - a(\tau)_j) \cdot f_j + \sum_{j=k^*+1}^m a^*_j \cdot f_j\,.
	\end{align}
	Note that the definition of $k^*$ implies $q(\tau)_{k^*} = Q'_{k^*}$. Moreover, since Algorithm~\ref{alg:opt-alg} uses protection levels $Q'$ till the start of time $\tau$, we have $q(\tau)_j \leq Q'_j$ for all $j \in [m]$. Therefore, $q(\tau)_{k^*} - q(\tau)_{p-1} \geq Q'_{k^*} - Q'_{p-1}$ for all $p \leq k^*$ and consequently, all four terms in \eqref{eqn:comp-case-3-inter-1} are non-negative.

	Next we bound $Opt(I)$ from above to show $Alg(I) \geq\gamma \cdot Opt(I)$. First, observe that if $i \leq k^* \leq \ell$, then we have $Opt(I(A, k^*) \oplus I(F,i)) = n \cdot f_{k^*}$ and consequently
			\begin{align*}
				Opt(I) &\leq \left(n - \sum_{j=k^*}^m a_j^* \right) \cdot f_{k^*} + \sum_{j= k^*+1}^m a_j^* \cdot (f_j - f_{k^*})\\
				&= n \cdot f_{k^*} + \sum_{j= k^*+1}^m a_j^* \cdot (f_j - f_{k^*})\\
				&= Opt(I(A, k^*) \oplus I(F,i)) + \sum_{j= k^*+1}^m a_j^* \cdot (f_j - f_{k^*})\\
			\end{align*}
	Combining this with \eqref{eqn:comp-case-3-inter-1} yields $Alg(I) \geq\gamma \cdot Opt(I)$ as required.
	
	Next, suppose $k^* \geq \ell$. Then, for $i' = \max\{i, \ell\}$, $Opt(I)$ is composed of at most $a_j^*$ customers with fare $f_j$ for $j \geq i'$ and the remainder with fares less than or equal to $f_{i'}$. Consequently, we get:
	\begin{align}\label{eqn:comp-case-3-inter-2}
		Opt(I) &\leq n \cdot f_{i'} + \sum_{j=i'+1}^m a_j^* \cdot (f_j - f_{i'}) \nonumber \\
		&= n \cdot f_{i'} + \sum_{j=i'+1}^{k^*} a(\tau)_j \cdot (f_j - f_{i'}) + \sum_{j= i'+1}^{k^*} (a_j^* - a(\tau)_j) \cdot (f_j - f_{i'})  + \sum_{j=k^*+1}^m a_j^* \cdot (f_j - f_{i'}) \nonumber \\
		&\leq \underbrace{n \cdot f_{i'} + \sum_{j=i'+1}^{k^*} a(\tau)_j \cdot (f_j - f_{i'})}_{\clubsuit} + \sum_{j= i+1}^{k^*} (a_j^* - a(\tau)_j) \cdot f_j  + \sum_{j=k^*+1}^m a_j^* \cdot f_j \,,
	\end{align}
	where the last inequality follows from $i \leq i' = \max\{i, \ell\}$. Since $\delta = 0$ till $t = \tau-1$, the definition of $h(\tau)$ implies that $\sum_{g=j}^{k^*} r_g(\tau) \leq \sum_{g=j}^{k^*} \{A_g - x_g\}$ for all $j \in [\ell+1, k^*]$. Moreover, $a(\tau)_j = q(\tau)_j - q(\tau)_{j-1} + r(\tau)_j$ because the total amount of customers accepted by the algorithm and the total amount of customers rejected by the algorithm together equal the total number of customers that arrived. This allows us to write
	\begin{align}
		\clubsuit &= n \cdot f_{i'} + \sum_{j=i'+1}^{k^*} r(\tau)_j \cdot (f_j - f_{i'}) + \sum_{j=i'+1}^{k^*} (q(\tau)_j - q(\tau)_{j-1}) \cdot (f_j - f_{i'}) \nonumber \\
		&\leq n \cdot f_{i'} + \sum_{j=i'+1}^{k^*} (A_j - x_j) \cdot (f_j - f_{i'}) + \sum_{j=i'+1}^{k^*} (q(\tau)_j - q(\tau)_{j-1}) \cdot (f_j - f_{i'}) \nonumber \\
		&= n \cdot f_{i'} + \sum_{j=i'+1}^{k^*} A_j \cdot (f_j - f_{i'})+ \sum_{j=i'+1}^{k^*} (q(\tau)_j - q(\tau)_{j-1} - x_j) \cdot (f_j - f_{i'}) \nonumber \\
		&= \left(n - \sum_{j=i'+1}^{k^*} A_j \right) \cdot f_{i'} + \sum_{j=i'+1}^{k^*} A_j \cdot f_j + \sum_{j=i'+1}^{k^*} (q(\tau)_j - q(\tau)_{j-1} - x_j) \cdot (f_j - f_{i'}) \nonumber \\
		&= Opt(I(A,k^*) \oplus I(F,i')) + \sum_{j=i'+1}^{k^*} (q(\tau)_j - q(\tau)_{j-1} - x_j) \cdot (f_j - f_{i'}) \nonumber \\
		&= Opt(I(A,k^*) \oplus I(F,i')) + \sum_{p=i'+1}^{k^*} \left(q(\tau)_{k^*} - q(\tau)_{p-1} - \sum_{g=p}^{k^*} x_g \right) \cdot (f_p - f_{p-1}) \label{eqn:clubsuit-inter-1}\\
		&\leq Opt(I(A,k^*) \oplus I(F,i)) + \sum_{p=i'+1}^{k^*} \left(q(\tau)_{k^*} - q(\tau)_{p-1} - (Q_{k^*}' - Q'_{p-1}) \right) \cdot (f_p - f_{p-1}) \label{eqn:comp-case-3-inter-3}
	\end{align}
	\eqref{eqn:clubsuit-inter-1} follows from Lemma~\ref{lemma:rev-rewrite}. \eqref{eqn:comp-case-3-inter-3} follows from the observation that $Opt(I(A,k^*) \oplus I(F,\ell)) = Opt(I(A,k^*)) \leq Opt(I(A,k^*) \oplus I(F,i))$ when $\ell \leq k^*$.
	
	Finally, combining \eqref{eqn:comp-case-3-inter-1}, \eqref{eqn:comp-case-3-inter-2}, \eqref{eqn:comp-case-3-inter-3} and using the fact that $\gamma \leq 1$ yields
	\begin{align*}
		Alg(I) - \gamma \cdot Opt(I) \geq &\sum_{p=i+1}^{k^*} (q(\tau)_{k^*} - q(\tau)_{p-1} - (Q'_{k^*} - Q'_{p-1})) (f_p - f_{p-1})\\
		&- \gamma \cdot \sum_{p=i'+1}^{k^*} \left(q(\tau)_{k^*} - q(\tau)_{p-1} - (Q_{k^*}' - Q'_{p-1}) \right) \cdot (f_p - f_{p-1})\\
		\geq &\ 0
	\end{align*}
	where the last inequality follows from the observation that $q(\tau)_{k^*} = Q_{k^*}'$ and $q(\tau)_j \leq Q'_j$ for all $j \in [m]$ and $i \leq i' = \max\{i, \ell\}$. Hence, we have shown that Algorithm~\ref{alg:opt-alg} achieves revenue greater than or equal to $\gamma \cdot Opt(I)$ in this case.
\end{proof}

We now have all the ingredients in place to prove Theorem~\ref{thm:lower-bound}.


\begin{proof}[Proof of Theorem~\ref{thm:lower-bound}]
	In Theorem~\ref{thm:upper-bound}, we established the following upper bound: $\beta(A, \gamma) \leq \beta^*$. In Lemma~\ref{lemma:alg-consistent}, we established that Algorithm~\ref{alg:opt-alg} is $\beta^*$-consistent on advice $A$. Furthermore, through Lemma~\ref{lemma:alg-prefix} and Lemma~\ref{lemma:alg-competitive}, we established the $\gamma$-competitiveness of Algorithm~\ref{alg:opt-alg}. Hence, from the definition of $\beta(A, \gamma)$, we get $\beta^* \leq \beta(A,\gamma)$, thereby implying $\beta^* = \beta(A,\gamma)$.
\end{proof}

}

\section{Missing Proofs from Section~\ref{sec:protec}}\label{sec:appendix:protec}

\subsection{Proof of Theorem~\ref{thm:opt-protec}}\label{appendix:protec}

In this section, we prove that the protection level policy based on the protection levels returned by Algorithm~\ref{alg:opt-protec} is $\gamma$-competitive and $(\beta^{PL}(A, \gamma) - \epsilon)$-consistent on advice $A$. To do so, we show that, for $\beta \in [0,1]$, the protection levels returned by CoreSubroutine$(\beta)$ satisfy $Q_m \leq n$ if and only if $\beta \leq \beta^{PL}(A, \gamma)$. The following lemma takes a step in that direction by showing that, for a given $\beta \in [0,1]$, CoreSubroutine$(\beta)$ returns protection levels $Q$ such that $Q_m \leq n$, then the protection level policy based on $Q$ is $\beta$-consistent on advice $A$ while being $\gamma$-competitive. In other words, $Q_m \leq n$ implies $\beta \leq \beta^{PL}(A, \gamma)$).

\begin{lemma}\label{lemma:protec-first}
	Let $Q$ be the protection levels returned by CoreSubroutine$(\beta)$ for some $\beta \in [0,1]$. If $Q_m \leq n$, then $Q(I) \geq \beta \cdot Opt(A)$ for all $I \in S(A)$ and $Q(I) \geq \gamma \cdot Opt(I)$ for all instances $I \in \mathcal{U}$, i.e., the protection level policy corresponding to $Q$ is $\beta$-consistent on advice $A$ while also being $\gamma$-competitive.
\end{lemma}

\begin{proof}
	Let $Q$ be the protection levels returned by CoreSubroutine$(\beta)$ for some $\beta \in [0,1]$. Suppose $Q_m \leq n$. We will show that $Q(I) \geq \beta \cdot Opt(A)$ for all $I \in S(A)$ and $Q(I) \geq \gamma \cdot Opt(I)$ for all instances $I \in \mathcal{U}$. First, consider an instance $I \in S(A)$. Observe that, if $I'$ represents the instance obtained by reordering the customers in $I$ to arrive in increasing order of fares, then $Q(I') \leq Q(I)$ (Lemma~\ref{lemma:protec-inc-order}). Furthermore, since the top $n$ fares in $I'$ are the same as $I(A)$ and $Q_m \leq n$, we also have $Q(I(A)) \leq Q(I')$. Hence to show that $Q(I) \geq \beta \cdot Opt(A)$ for all $I \in S(A)$, it suffices to show $Q(I(A)) \geq \beta \cdot Opt(I(A))$, which we do next.
	
	Let $x_i$ represent the total quantity of customers with fare $f_i$ accepted by the protection level policy based on $Q$ when acting on the instance $I(A)$. If $x_i = N_i$ for all $i \in [m]$, then $Q(I(A)) \geq \beta \cdot Opt(I(A))$ trivially. Assume that this is not the case and let $j$ be the largest index such that $x_j < N_j$.  Then, from part (b) of the $k=j$ iteration of the for loop and the definition of $j$, we get that
	\begin{align*}
		Q(I(A)) = Q(I(A,j)) + \sum_{i=j+1}^m N_i f_i \geq \beta \cdot Opt(I(A))
	\end{align*}
	
	Next, to complete the proof we show that $Q(I) \geq \gamma \cdot Opt(I)$ for all $I \in \mathcal{U}$. To do so, we first show that $Q(I(F,k)) \geq \gamma \cdot Opt(I(F,k))$ for all $k \in [m]$. Since $Q_m \leq n$, we have $c_k + d_k \leq n$ for all $k \in [m]$, which further implies that $Q(I(F,k)) \geq c_k f_k + Q(I(F,k-1)) \geq \gamma \cdot Opt(I(F,k))$ due to the definition of $c_k$. Now, consider an arbitrary instance $I \in \mathcal{U}$ and let $q_j$ be the number of customers of fare type $j$ accepted when the protection level policy based on $Q$ acts on instance on $I$. Moreover, let $j$ be the smallest index such that $q_j = Q_j$, assuming $q_0 = Q_0 = 0$. Then, we get that the protection level policy based on $Q$ accepts every customer with fare greater than or equal to $f_{j+1}$ in $I$, and the revenue it attains from customers with fare less than or equal to $f_j$ is at least $\gamma \cdot Opt(I(F,j)) = \gamma \cdot n f_j$ because $Q(I(F, j)) \geq \gamma \cdot Opt(I(F,j))$.
	Hence, we have shown that $Q(I) \geq \gamma\cdot Opt(I)$, since $Opt(I) \leq Opt(I(F,j)) + R$, where $R$ is the revenue obtained from accepting all customer of fare type strictly greater than $j$ in $I$.
\end{proof}

The next lemma establishes that, if there exists a feasible protection level policy which is $\beta$-consistent on advice $A$ while being $\gamma$-competitive, then the protection levels $Q$ returned by CoreSubroutine$(\beta)$ satisfy $Q_m \leq n$. In other words, $\beta \leq \beta^{PL}(A, \gamma)$ implies $Q_m \leq n$.

\begin{lemma}\label{lemma:protec-second}
	Let $Q$ be the protection levels returned by CoreSubroutine$(\beta)$ for some $\beta \in [0,1]$. If there exist feasible protection levels $Q^*$ such that $Q^*(I) \geq \beta \cdot Opt(A)$ for all $I \in S(A)$ and $Q^*(I) \geq \gamma \cdot Opt(I)$ for all instances $I \in \mathcal{U}$, i.e., the protection level policy corresponding to $Q^*$ is $\beta$-consistent on advice $A$ while also being $\gamma$-competitive, then $Q_m \leq n$.
\end{lemma}

\begin{proof}
	Suppose there exist feasible protection levels $Q^*$ such that $Q^*(I) \geq \beta \cdot Opt(A)$ for all $I \in S(A)$ and $Q^*(I) \geq \gamma \cdot Opt(I)$ for all instances $I \in \mathcal{U}$. To prove $Q_m \leq n$, it suffices to show that $Q_m \leq Q_m^*$.
	
	First, we define a `For' loop to inductively define auxiliary variables $\{y(k)_i, z(k)_i\}_{1 \leq i \leq k \leq m}$ and protection levels $Q'$, all initialized to be zero, which will play a central role in the proof.
	
	For $k = 1$ to $m$:
	\begin{itemize}
		\item[(a)] Set $\{y(k)_i\}_{i\leq k}$ to be an optimal solution to the following optimization problem
			\begin{align*}
				\max \quad  & \sum_{i=1}^k y(k)_i\\
				\text{s.t.} \quad &	Q_i' - Q_{i-1}' + y(k)_i \leq Q^*_i - Q^*_{i-1} \quad \forall i \leq k\\
				& \sum_{i=1}^k y(k)_i \leq c_k
			\end{align*}
		
			Set $Q_i' \leftarrow Q'_i + \sum_{r=1}^i y(k)_r$ for all $i \leq k$ and $Q'_i = Q'_k$ for all $i > k$.
			
		\item[(b)] Set $\{z(k)_i\}_{i\leq k}$ to be an optimal solution to the following optimization problem
			\begin{align*}
				\max \quad  & \sum_{i=1}^k z(k)_i\\
				\text{s.t.} \quad &	Q_i' - Q_{i-1}' + z(k)_i \leq Q^*_i - Q^*_{i-1} \quad \forall i \leq k\\
				& \sum_{i=1}^k z(k)_i \leq d_k
			\end{align*}
		
			Set $Q_i' \leftarrow Q'_i + \sum_{r=1}^i z(k)_r$ for all $i \leq k$ and $Q'_i = Q'_k$ for all $i > k$.
		
%
	\end{itemize}
	
	Next, we use induction on $k \in [m]$ to show that $\sum_{i=1}^k y(k)_i = c_k$ and $\sum_{i=1}^k z(k)_i = d_k$ for all $k \in [m]$. First, consider the base case $k = 1$. Since $Q^*$ satisfies $Q^*(I(F,1))\geq \gamma\cdot Opt(I(F,1))$, we have $y(1)_1 = c_1$ because $c_1 \leq Q^*_1$.
	Moreover, since $Q^*$ satisfies $Q^*(I(A,1)) + \sum_{i=2}^m N_i f_i \geq \beta \cdot Opt(I(A))$, we have $z(1)_1 = d_1$ because $d_1 f_1 = (\beta \cdot Opt(I(A)) - \sum_{i=2}^m N_i f_i) - c_1f_1 \leq Q^*(I(A,1)) - c_1f_1$ implies
	$d_1 \leq Q^*_1 - c_1$.
	Hence, we have established the base case $k=1$ for the induction.

	Assume the induction hypothesis holds, i.e., suppose $\sum_{i=1}^{p} y(p)_i = c_{p}$ and $\sum_{i=1}^{p} z(p)_i = d_{p}$ for all $p < k$, where $2\leq k \leq m$. Consider iteration $k$ of the two `For' loops (the one in CoreSubroutine$(\beta)$ and the one defined above). Since $f_1 < \dots f_m$, at the beginning of step (a), we have
	\begin{align*}
		Q'(I(F,k-1)) = \sum_{i=1}^{k-1} (Q_I' - Q_{i-1}') f_i = \sum_{i=1}^{k-1} \sum_{p=i}^{k-1} [y(p)_i + z(p)_i] f_i \leq \sum_{p=1}^{k-1} (c_p + d_p) f_p = Q(I(F,k-1))
	\end{align*}
	Next, note that, at the beginning of step (a), $Q^*$ and $Q'$ satisfy
	\begin{align*}
		\sum_{i=1}^k [(Q^*_i - Q^*_{i-1}) - (Q'_{i} - Q'_{i-1})]  f_i + Q'(I(F,k-1)) &= \sum_{i=1}^k (Q^*_i - Q^*_{i-1}) f_i\\
		 &= Q^*(I(F,k))\\
		 &\geq \gamma\cdot Opt(I(F,k))\\
		 &= Q(I(F,k-1)) + c_k f_k
	\end{align*}
	Combining the above inequalities yields $\sum_{i=1}^k [(Q^*_i - Q^*_{i-1}) - (Q'_{i} - Q'_{i-1})]  f_i \geq c_k f_k$, which further implies $\sum_{i=1}^k [(Q^*_i - Q^*_{i-1}) - (Q'_{i} - Q'_{i-1})] \geq c_k$ at the beginning of step~(a). As a consequence, we get that any optimal $\{y(k)_i\}_{i\leq k}$ in step (a) satisfies $\sum_{i=1}^k y(k)_i = c_k$.
	
	If $d_k = 0$, then $\sum_{i=1}^k z(k)_i = 0$, as required. Consider the case when $d_k > 0$. Let $Q^{[k]}$ denote the protection levels $Q$, as defined in Algorithm~\ref{alg:opt-protec}, at the end of step~(a) of iteration $k$; i.e.,
	\begin{align*}
		Q^{[k]}_i - Q^{[k]}_{i-1} =
		\begin{cases}
			c_i + d_i &\text{ if } i < k\\
			c_k &\text{ if } i = k\\
			0 &\text{ if } i> k
		\end{cases}
	\end{align*}
	Moreover, let $x_i$ denote the number of customers with fare $f_i$ that the protection level policy based on $Q^{[k]}$ accepts when acting on the instance $I(A,k)$. We would like to show $x_k < N_k$. For contradiction, suppose $x_k = N_k$. If $x_i = N_i$ for all $i \leq k$, then $Q^{[k]}(I(A,k)) + \sum_{i= k+1}^m N_i f_i = \sum_{i=1}^m N_i f_i \geq \beta \cdot Opt(I(A))$, thereby implying $d_k = 0$, which is a contradiction. Now, assume that there exists $i$ such that $x_i < N_i$ (and still assuming $x_k=N_k$), and let $j = \max\{i < k \mid x_i < N_i\}$. Then, from step~(b) of the $j$-th iteration and the definition of $j$, we get $Q^{[k]}(I(A,k)) + \sum_{i= k+1}^m N_i f_i = Q^{[k]}(I(A,j)) + \sum_{i=j+1}^m N_i f_i = \beta \cdot Opt(I(A))$, which contradicts $d_k > 0$. Hence, we have shown that $x_k < N_k$.
	
	Additionally, due to the induction hypothesis and $\sum_{i=1}^k y(k)_i = c_k$, at the end of step~(a) of the $k$-th iteration, we have $Q'_k - Q'_{i-1} \leq Q^{[k]}_k - Q^{[k]}_{i-1}$ for all $i \leq k$. Here $Q'_k -Q'_{i-1}$ and $Q^{[k]}_k - Q^{[k]}_{i-1}$ represent the number of customers with fare in the set $\{f_i, \dots f_k\}$ accepted by $Q'$ and $Q^{[k]}$ respectively. Combining this with $x_k < N_k$ and $d_k > 0$ yields
	\begin{align*}
		Q'(I(A,k)) \leq Q^{[k]}(I(A,k)) < \beta \cdot Opt(I(A)) - \sum_{i=k+1}^m N_i f_i = Q^{[k]}(I(A,k)) + d_k f_k
	\end{align*}
	at the end of step~(a) of iteration $k$. Moreover, since $Q^*$ is $\beta$-consistent on advice $A$, we also have $ \beta \cdot Opt(I(A)) - \sum_{i=k+1}^m N_i f_i \leq Q^*(I(A,k))$. Therefore, at the beginning of step~(b) of iteration $k$, we have $Q^*(I(A,k)) - Q'(I(A,k)) \geq d_k f_k$.
	
	For contradiction, suppose $\sum_{i=1}^k z(k)_i < d_k$. Then, at the beginning of step~(b) of iteration $k$, we have $Q_i' - Q_{i-1}' + z(k)_i = Q^*_i - Q^*_{i-1}$ for all $i \leq k$, which yields
	\begin{align*}
		Q^*(I(A,k)) - Q'(I(A,k)) \leq \left(\sum_{i=1}^k z(k)_i \right) f_k.
	\end{align*}
	Combining this with $Q^*(I(A,k)) - Q'(I(A,k)) \geq d_k f_k$, we get $d_kf_k \leq \left(\sum_{i=1}^k z(k)_i \right) f_k$,
	which contradicts $\sum_{i=1}^k z(k)_i < d_k$. Hence, we have shown that $\sum_{i=1}^k z(k)_i < d_k$, thereby completing the induction step.
	
	The lemma follows as a consequence because, at the end of iteration $m$, we have
	\begin{align*}
		Q_m = \sum_{k=1}^m c_k + d_k = \sum_{k=1}^m \sum_{i=1}^k y(k)_i + z(k)_i = \sum_{i=1}^m \sum_{k= i}^m y(k)_i + z(k)_i = \sum_{i=1}^m Q'_{i}- Q'_{i-1} \leq \sum_{i=1}^m Q^*_{i}- Q^*_{i-1} =  Q_m^*
	\end{align*}
	where the inequality follows from the second constraint of the optimization problem in step~(b) of iteration $m$ and the update rule for $Q'$.
\end{proof}

Combining Lemma~\ref{lemma:protec-first} and Lemma~\ref{lemma:protec-second} allows us to prove Theorem~\ref{thm:opt-protec}:


\begin{proof}[Proof of Theorem~\ref{thm:opt-protec}]
	At termination of the `While' loop of Algorithm~\ref{alg:opt-protec}, we have $\bar \beta - \epsilon \leq \underline \beta \leq \bar \beta$. Moreover, as a consequence of Lemma~\ref{lemma:protec-first}, the protection levels $Q$ returned by Algorithm~\ref{alg:opt-protec} are $\underline \beta$-consistent on advice $A$ and $\gamma$-competitive. If $\bar \beta = 1$, then the theorem holds. Assume $\bar \beta < 1$. Then, Lemma~\ref{lemma:protec-second} and step~3 of Algorithm~\ref{alg:opt-protec} imply that no protection level policy is $\bar \beta$-consistent on advice $A$ while being $\gamma$-competitive. In other words $\beta^{PL}(A, \gamma) \leq \bar \beta$. As $\bar \beta - \epsilon \leq \underline \beta$, we get that the protection levels $Q$ returned by Algorithm~\ref{alg:opt-protec} are $(\beta^{PL}(A, \gamma) - \epsilon)$-consistent on advice $A$ and $\gamma$-competitive. Finally, note that there are at most $\log(1/\epsilon)$ iterations of the While loop, and in each iteration of the While loop there is 1 call to CoreSubroutine. Since computing $Q(I)$ takes linear time for every instance $I$ and protection levels $Q$, each iteration of the `For' loop in CoreSubroutine runs in polynomial time. Combining the aforementioned runtimes yields the polynomial runtime of Algorithm~\ref{alg:opt-protec}.
\end{proof}

%
%
%
%
%

\revision{
\section{Missing Proofs from Section~\ref{sec:robustness}}\label{appendix:robustness}

\begin{figure}[t!]
\begin{algorithm}[H]
   \caption{$\Lambda$-relaxed Algorithm~\ref{alg:opt-alg}}
   \label{alg:robust-opt-alg}
    \begin{algorithmic}\vspace{0.08cm}
    		\item[\textbf{Input:}] Required level of competitiveness $\gamma \in [0,c(F)]$, advice $A$ and instance $I = \{s_t\}_{t=1}^T$.
    		\vspace{0.5em}
    		\item[$\mathbf{t = 0}$:] Solve \ref{LP} to find optimal solution $(\beta^*, x, (y(k))_k)$. Define protection levels:
    			\begin{itemize}
    				\item For $i \in [m]$, set $Q'_i = \sum_{j=1}^i \lfloor x_j \rfloor$. Moreover, set $Q_0' = 0$.
					\item For $i,k \in [m]$, set $R(k)_i = Q'_i + \sum_{j=1}^i y(k)_j$ if $i \leq k$ and $R(k)_i = Q'_k+ \sum_{j=1}^i y(k)_j$ if $i > k$. Moreover, set $R(k)_0 = 0$ for all $k \in [m]$.
    			\end{itemize}
    		\item[\textbf{Initialize:}] Total accepted customers with fare $f_j$ or below as $q_j = 0$ for all $j \in [m]$; total rejected customers with fare $f_j$ as $r_j = 0$ for all $j \in [m]$; trigger $\delta = 0$; and active protection levels $Q^*_i = Q'_i$ for all $i \in [m]$. Rejection buffer $B = 0$
    		\vspace{0.5em}
            \item[\textbf{For}] $t=1$ to $T$:
            \begin{enumerate}
                \item Let $p \in [m]$ be the fare class of customer $s_t$, i.e., $s_t = f_p$. Calculate the fraction of customer $s_t$ that would be accepted under $Q^*$: set $w^* = \max\{w \in [0,1] \mid q_j + w \leq Q_j^*\ \forall j \geq p\}$ and let $h \geq p$ be the largest index such that $q_p + w^* = Q^*_j$ (set $h= 0$ if no such index exists). If $\delta = 0$, update $r_p \leftarrow r_p + (1 - w^*)$ to reflect the anticipated rejection of a $1 - w^*$ fraction of $s_t$.

                \item \emph{Check Trigger Condition.} If $\delta = 0$ and $\sum_{i=j}^h r_i > \Lambda + \sum_{i=j}^h \{A_i - x_i\}$ for some $j \geq \ell + 1$, i.e., the total rejections among fares $\{f_j, \dots, f_h \}$ exceeds the rejections among fares $\{f_j, \dots, f_m\}$ in $I(A)$ under LP solution $(\beta^*, x, (y(k))_k)$:
                	\begin{itemize}
                		\item Set trigger $\delta = 1$ and $k = \max\{j \mid q_j = Q_j'\}$ (set $k = 1$ if no such index exists).
                		\item \textbf{While} $\exists\ j \in [m]$ such that $R(k)_j < q_j$: Set $k = \min\{j \in [m] \mid R(k)_j < q_j\}$.\\
                			Set $k^* = k$, i.e., $k^*$ is the value with which the `While' loop terminates.
                		\item \emph{Switch Protection Levels.} Set $Q^*_j = R(k^*)_j$ for all $j \in [m]$.
                	\end{itemize}

                \item \emph{Make decision according to $Q^*$.} Accept $w^* = \max\{w \in [0,1] \mid q_j + w \leq Q_j^*\ \forall j \geq p\}$
                 fraction of customer $s_t$ and update $q_j \leftarrow q_j + w$, for all $j \geq p$, to reflect the increase in number of customers of fare $f_j$ or lower accepted by the algorithm.

            \end{enumerate}
    \end{algorithmic}
 \end{algorithm}
\end{figure}

\subsection{Proof of Theorem~\ref{thm:protec-robust}}

\begin{proof}[Proof of Theorem~\ref{thm:protec-robust}]
	Using Lemma~\ref{lemma:protec-inc-order}, we can assume without loss of generality that the customers in $I$ arrive in increasing order of fares. Let $q_j^*$ be the total number of customers with fare $f_j$ or less accepted by the protection levels $Q$. Moreover, let $i$ be the largest index such that $q_i^* = Q_i$. If not such index exists, then the protection level policy based on $Q$ accepts every customer and the lemma follows directly. Since $q_j^* \leq Q_j$ for all $j \leq i$, we have $q_i^* - q^*_{j-1} \geq Q_i - Q_{j-1}$ for all $j \leq i$. Moreover, note that $Q$ accept every customer with fare strictly greater than $f_i$ and accepts $q_j^* - q_{j-1}^*$ customers with fare $f_j$ for all $j \leq i$. Consider the following two cases,
	\begin{itemize}
		\item $\boldsymbol{i < \ell}$: Then, Algorithm~\ref{alg:robust-opt-alg} accepts every customer with fare $f_\ell$ or larger. Therefore,
			\begin{align*}
				Q(I) &\geq \sum_{j=\ell}^m K_I(j) \cdot f_j\\
				&\geq \sum_{j=\ell}^m \left(A_j - (A_j - K_I(j))^+\right) \cdot f_j \\
				&\geq  Opt(I(A)) - f_m \cdot \lambda(I,A)\\
				&\geq Opt(I) - 2 f_m \cdot \lambda(I,A)\,.
			\end{align*}
			
		\item $\boldsymbol{i \geq \ell}$: Then, we can write
			\begin{align*}
				Q(I) &= \sum_{j=1}^i (q_j^* - q_{j-1}^*) \cdot f_j + \sum_{j=i+1}^m K_I(j) \cdot f_j\\
				&= \sum_{j=1}^i (q_i^* - q^*_{j-1}) \cdot (f_j - f_{j-1}) + \sum_{j=i+1}^m K_I(j) \cdot f_j \tag{Lemma~\ref{lemma:rev-rewrite}}\\
				&\geq \sum_{j=1}^i (Q_i - Q_{j-1}) \cdot (f_j - f_{j-1}) + \sum_{j=i+1}^m K_I(j) \cdot f_j\\
				&\geq \sum_{j=1}^i (Q_i - Q_{j-1}) \cdot (f_j - f_{j-1}) + \sum_{j=i+1}^m A_j \cdot f_j - \sum_{j=i+1}^m |A_j - K_I(j)| \cdot f_j \\
				&= \sum_{j=1}^i (Q_j - Q_{j-1}) \cdot f_j + \sum_{j=i+1}^m A_j \cdot f_j - \sum_{j=i+1}^m |A_j - K_I(j)| \cdot f_j \tag{Lemma~\ref{lemma:rev-rewrite}}\\
				&\geq Q(I(A)) - f_m \cdot \lambda(I,A)\\
				&\geq \beta(A,\gamma) \cdot Opt(I(A)) - m \cdot \lambda(I,A) \cdot f_m\\
				& \geq \beta(A,\gamma) \cdot Opt(I) - 2\cdot \lambda(I,A) \cdot f_m\,. \qedhere
			\end{align*}
	\end{itemize}
\end{proof}

\subsection{Proof of Theorem~\ref{thm:robust-opt}}

We begin by showing that the trigger of Algorithm~\ref{alg:robust-opt-alg} is not invoked whenever $\lambda(I,A) \leq \Lambda$.

\begin{lemma}\label{lemma:delayed-trigger}
	On any instance $I = \{s_t\}_{t=1}^T$ such that $\lambda(I,A) \leq \Lambda$, we have $\delta =0$ throughout the run of Algorithm~\ref{alg:robust-opt-alg}.
\end{lemma}
\begin{proof}
	For contradiction, let $\tau \leq T$ be the first time step when $\delta = 1$. Let $s_\tau = f_p$ for some $p \geq \ell+1$. Moreover, let $q(\tau)_j$ be the total number of customers with fare $f_j$ or below accepted before the start of time step $\tau$ and $r(\tau)_j$ denote the total number of customers with fare $f_j$ rejected before the start of time step $\tau$. For $w^* = \max\{w \in [0,1] \mid q(\tau)_j + w \leq Q_j'\ \forall j \geq p\}$, let $h \geq p$ be the largest index such that $q(\tau)_h + w^* = Q'_h$. Such an index always exists because $w^*< 1$ by definition of $\tau$. Moreover, the definition of $\tau$ also implies the existence of an index $g \in [p, h]$ such that $(1 - w^*) + \sum_{i=g}^h r(\tau)_i > \Lambda + \sum_{i=g}^h \{A_i - x_i\}$. Since $q(\tau)_{g-1} \leq Q'_{g-1}$ by definition of protection level policies, we get that the algorithm accepted at least $q(\tau)_h - q(\tau)_{g-1} \geq Q'_h - w^* - Q'_{g-1}$ customers with fares $\{f_g ,\dots, f_h\}$ before time $\tau$. Since the $\tau$-th customer has fare $f_p$ such that $g \leq p \leq h$, we get that the total number of customers with fare $\{f_g ,\dots, f_h\}$ that arrive in $I$ is at least $1 + Q'_h - w^* - Q'_{g-1} + \sum_{i=g}^h r(\tau)_i$. Moreover, using $Q'_h - Q'_{g=1} = \sum_{i=g}^h x_i$, we can bound this from below:
	\begin{align*}
		Q'_h - Q'_{g-1} + (1 - w^*) + \sum_{i=g}^h r(\tau)_i > \Lambda + \sum_{i=g}^h x_i + \sum_{i=g}^h \{A_i - x_i\} = \Lambda + \sum_{i=g}^h A_i \,,
	\end{align*}
	which contradicts $\lambda(I,A) \leq \Lambda$. Therefore, we have $\delta = 0$ throughout the run of Algorithm~\ref{alg:robust-opt-alg}.
\end{proof}

\begin{proof}[Proof of Theorem~\ref{thm:robust-opt}]
	Fix advice $A$, target level of competitiveness $\gamma \in c(F)$ and $\Lambda > 0$.
	
	\begin{enumerate}
		\item Consider an instance $I$ such that $\lambda(I,A) \leq \Lambda$. Then, Lemma~\ref{lemma:delayed-trigger} implies that $\delta = 0$ throughout the run of Algorithm~\ref{alg:robust-opt-alg} and Algorithm~\ref{alg:robust-opt-alg} makes the same decisions as the protection level policy based on $Q'$. This allows us to employ Theorem~\ref{thm:protec-robust} to get the desired revenue bound:
			\begin{align*}
				Q'(I) \geq \frac{Q'(I(A))}{Opt(I(A))} \cdot Opt(I) - 2f_m \cdot \lambda(I,A) = \beta(A, \gamma) \cdot Opt(I) - 2f_m \cdot \lambda(I,A)
			\end{align*}

		\item Consider an arbitrary instance $I = \{s_t\}_{t=1}^T \in \mathcal{U}$. We will first prove the result under the assumption that $x_i \in \mathbb Z$ for all $i \in [n]$, and then show that this assumption is without much loss. Assume $x_i \in \mathbb Z$ for all $i \in [n]$. Therefore, when $\delta = 0$, Algorithm~\ref{alg:opt-alg} and Algorithm~\ref{alg:robust-opt-alg} either completely accept a customer or completely reject her. Let $t^*$ be the first time at which $\sum_{i=j}^h r_i > \Lambda + \sum_{i=j}^h \{A_i - x_i\}$ for some $j \geq \ell + 1$, and let $\tau$ be the first time at which $\sum_{i=j}^h r_i > \sum_{i=j}^h \{A_i - x_i\}$ for some $j \geq \ell + 1$.  Construct instance $I'$ from $I$ as follows: Initialize $I'_{t^* -1}$ be the instance composed of customers $\{s_t\}_{t=1}^{t^*-1}$. For $t \in [t^*, \tau)$, inductively define $I'_t$ to be equal to
			\begin{itemize}
				\item $I'_{t-1} \oplus \{s_t\}$ if $\delta = 0$ throughout the run of Algorithm~\ref{alg:opt-alg} on $I'_{t-1} \oplus \{s_t\}$;
				\item $I'_{t-1}$ otherwise.
			\end{itemize}
		 Set $I' \coloneqq I'_{\tau-1} \oplus \{s_t\}_{t=\tau}^T$. Since the trigger is invoked based only on rejections, $I\setminus I'$ is composed purely of customers that were rejected by $Q'$. Moreover, the set customers accepted by Algorithm~\ref{alg:opt-alg} in $I'$ is the same as the set of customers accepted by Algorithm~\ref{alg:robust-opt-alg} in $I$ till time $\tau - 1$, both being equal to the set of customers accepted by $Q'$. Finally, the trigger of Algorithm~\ref{alg:opt-alg} is also invoked at time $\tau$ on instance $I'$. Therefore, the revenue of the $\Lambda$-relaxed Algorithm~\ref{alg:opt-alg} on instance $I$ is the same as the revenue of Algorithm~\ref{alg:opt-alg} (the unrelaxed version) on instance $I'$. Moreover, note that $Opt(I') \geq Opt(I) - f_m \cdot \Lambda$ because $|I\setminus I'| \leq \Lambda$. Hence, Algorithm~\ref{alg:robust-opt-alg} earns a revenue of at least $\gamma \cdot Opt(I) - f_m \cdot \Lambda$.
		
		 We conclude by showing that our assumption on the integrality of $\{x_i\}$ was without much loss. It is easy to see that replacing $x_i$ with $\lfloor x_i \rfloor$ only reduces the revenue of Algorithm~\ref{alg:opt-alg} by at most $m \cdot f_m$ on any instance. Hence, Algorithm~\ref{alg:robust-opt-alg} earns a revenue of at least $\gamma \cdot Opt(I) - f_m \cdot \Lambda - f_m \cdot m$. \qedhere
	\end{enumerate}
\end{proof}
}

%
%
%

\end{document}